\documentclass[journal,onecolumn,draftclsnofoot]{IEEEtran}
\IEEEoverridecommandlockouts
\usepackage{amsmath,amsfonts,amsthm,amssymb,amsbsy,bm,paralist,color}	
\usepackage{graphicx,algorithmic,algorithm}
\usepackage{physics}
\usepackage{cite}
\usepackage{xcolor}
\usepackage{cases}
\usepackage{colortbl}
\usepackage{arydshln}
\usepackage{caption}
\usepackage{enumitem}
\usepackage{algorithmic}

\theoremstyle{definition}
\newtheorem{lem}{Lemma}[]
\newtheorem{thm}{Theorem}[]
\newtheorem{prop}{Proposition}[]
\theoremstyle{definition}

\theoremstyle{remark}
\newtheorem{remark}{Remark}
\newcommand\peqref[1]{(\mathcal{P}\ref{#1})}
\newcommand\fref[1]{Fig. \ref{#1}}      
\newcounter{prb}
\definecolor{orange}{RGB}{255,107,0}
\definecolor{green}{RGB}{0,160,20}

\newcommand{\onesn}{\ensuremath{\bm{1}_{n,\phi}}}

\begin{document}

		\title{ DRL-Assisted Resource Allocation for NOMA-MEC Offloading with Hybrid SIC  }  
	\author{Haodong Li, Fang Fang, Zhiguo Ding
		\thanks{H. Li and Z. Ding are with the Department of Electrical and Electronic Engineering, The University of Manchester, M13 9PL, UK. F. Fang is with the Department of Engineering, Durham University, Durham DH1 3LE, UK (e-mail: haodong.li@manchester.ac.uk, fang.fang@durham.ac.uk, zhiguo.ding@manchester.ac.uk).}}
	\maketitle
\begin{abstract}
	Multi-access edge computing (MEC) and non-orthogonal multiple access (NOMA) have been regarded as promising technologies to improve computation capability and offloading efficiency of the mobile devices in the sixth generation (6G) mobile system. This paper mainly focuses on the hybrid NOMA-MEC system, where multiple users are first grouped into pairs, and users in each pair offload their tasks simultaneously by NOMA, and then a
	dedicated time duration is scheduled to the more delay-tolerable user for uploading the remaining data by orthogonal
	multiple access (OMA). For the conventional NOMA uplink transmission, successive interference cancellation (SIC) is
	applied to decode the superposed signals successively according to the channel state information (CSI) or the quality
	of service (QoS) requirement. In this work, we integrate the hybrid SIC scheme which dynamically adapts the SIC
	decoding order among all NOMA groups. To solve the user grouping problem, a deep reinforcement learning
	(DRL) based algorithm is proposed to obtain a close-to-optimal user grouping policy. Moreover, we optimally
	minimize the offloading energy consumption by obtaining the closed-form solution to the resource allocation problem.
	Simulation results show that the proposed algorithm converges fast, and the NOMA-MEC scheme outperforms the existing orthogonal multiple access (OMA) scheme.

\end{abstract}

\section{Introduction}

With fifth-generation (5G) networks being available now, the sixth-generation (6G) wireless network is currently under research, which is expected to provide superior performance to satisfy growing demands of mobile equipment, such as latency sensitive, energy hungry and computationally intensive services and applications \cite{MECNdu2020,9174795}. For example, the Internet of Things (IoT) networks are being developed rapidly, where massive numbers of nodes are supposed to be connected together, and IoT nodes can not only communicate with each others but also process acquired data \cite{ZHAO2020101184,9247446,9133107}. However, such IoT and many other terminal devices are constrained by the battery life and computational capability, and thereby these devices cannot support computationally intensive tasks. A conventional approach to improve the computation capability of mobile devices is mobile cloud Computing (MCC), where computation intensive tasks are offloaded to the central cloud servers for data processing \cite{DinhMCCsvy2013,8030322}. However, MCC will cause significant delays due to the long propagation distances. To address the offloading delay issue, especially for delay sensitive applications in the future 6G networks, multi-access edge computing (MEC) has been emerged as a decentralized structure to provide the computation capability close to the terminal devices, which are generally implemented at the base stations to provide cloud-like task processing service. \cite{8951269,9051691,8016573,8030322}.

From the communication perspective, non-orthogonal multiple access (NOMA) has been recognized as a promising technology to improve the spectral efficiency and massive connections, which enables multiple users to utilize the same resource block such as time and frequency for transmissions\cite{9358097,8972353}. Take power domain NOMA as an example, the signals of multiple users are multiplexed in power domain by the superposition coding, and at the receiver side, successive interference cancellation (SIC) is adopted to remove the multiple access interference successively\cite{8792153}. Hence, integrating NOMA with MEC can potentially improve the service quality of MEC including low transmission latency and massive connections compared to the conventional orthogonal multiple access (OMA).
\subsection{Related Works}
The integration of NOMA and MEC has been well studied so far, and researchers have proposed various approaches on optimal resource allocation to minimize the offloading delay and energy consumption. In \cite{9179779}, the author minimized the offloading latency for a multi-user scenario, in which the power allocation and task partition ratio were jointly optimized. The partial offloading policy can determine the amount of data to be offloaded to the server, and the remainder is processed locally. The author of \cite{8794550} proposed a iterative two-user NOMA scheme to minimize the offloading latency, in which two users offload their tasks simultaneously by NOMA. Since one of the users suffers performance degradation introduced by NOMA, instead of forcing two users to complete offloading at the same time, the remaining data is offloaded in together with the next user during the following time slot. Moreover, many existing works investigate the energy minimization of NOMA-MEC networks. For example, the joint optimization of central processing unit (CPU) frequency, task partition ratio and power allocation for a NOMA-MEC heterogeneous network were considered in \cite{8491299,9272879}. In \cite{8537962}, the author considered a multi-antenna NOMA-MEC network, and presented an approach to minimize the weighted sum energy consumption by jointly optimizing the computation and communication resource.

In addition to the existing works on pure NOMA schemes as aforementioned, a few works also combine NOMA and OMA in together, which is denominated as hybrid NOMA \cite{8673584}. In this paper, the author proposed a two-user hybrid NOMA scenario, in which one user is less delay tolerable than the other. The two users offload during the first time slot by NOMA, and the user with longer deadline offloads the remaining data during an additional time duration by OMA. This configuration presents significant benefits, which outperforms both OMA and pure NOMA in terms of energy consumption since the energy can be saved for the delay tolerable user instead of finishing offloading at the same time in pure NOMA networks. In \cite{9079198,LI2020241}, the hybrid NOMA scheme is extended to multi-user scenarios, in which a two-to-one matching algorithm is utilized to pair every two users into a group, and each group offload through a sub-carrier.

For the resource allocation in NOMA-MEC networks, user grouping is a non-convex problem, which is solved by exhaustive search or applying matching theory. Deep reinforcement learning (DRL) is recognized as a novel approach to this problem, which is a powerful tool to solve the real-time decision-making tasks, and only handful papers utilized it for user grouping and sub-channel assignment such as \cite{9044821,8790780} which output the user grouping policy for uplink and downlink NOMA networks respectively.

Moreover, in most of the NOMA works, the SIC decoding order is prefixed, which can either be determined by the channel state information (CSI) or the quality of service (QoS) requirements of users \cite{6868214,7273963,9146345}. A recent work \cite{9151196} has proposed a hybrid SIC scheme to switch the SIC decoding order dynamically, which has shown significant performance improvement in uplink NOMA networks. The author of \cite{9151208} integrated the hybrid SIC scheme with an MEC network to serve two uplink users, and the results reveals that the hybrid SIC outperforms the QoS based decoding order.

\subsection{Motivation and Contributions}
Motivated by the existing research on MEC-NOMA, in this paper, we investigate the energy minimization for the uplink transmission in multi-user hybrid NOMA-MEC networks with hybrid SIC. More specifically, a DRL based framework is proposed to generate a user grouping policy, and  the power allocation, time allocation and task partition assignment are jointly optimized for each group. The DRL framework collects experience data including CSI, deadlines, energy consumption as labeled data to train the neural networks (NNs). The main contributions of this paper are summarized as follows:
\begin{itemize}
	\item A hybrid NOMA-MEC network is proposed, in which an MEC server is deployed at the base station to serve multiple users. All users are divided into pairs, and each pair is assigned into one sub-channel. The users in each group adopt NOMA transmission with the hybrid SIC scheme in the first time duration, and the user with longer deadline transmits the remaining data by OMA in the following time duration. We propose a DRL-assisted user grouping framework with joint power allocation, time scheduling, and task partition assignment to minimize the offloading energy consumption under transmission latency and offloading data amount constraints.
	\item By assuming that the user grouping policy is given, the energy minimization problem for each group is non-convex due to the multiplications of variables and a 0-1 indicator function, which indicates two cases of decoding orders. The solution to the original problem can be obtained by solving each case separately. A multilevel programming method is proposed, where the energy minimization problem is decomposed into three sub-problems including power allocation, time scheduling, and task partition assignment. By carefully analyzing the convexity and monotonicity of each sub-problem, the solutions to all three sub-problems are obtained optimally in closed-form.The solution to the energy minimization problem for each case can be determined optimally by adopting the decisions successively from the lower level to the higher level (i.e., from the optimal task partition assignment to the optimal power allocation). Therefore, the solution to the original problem can be obtained by comparing the numerical results of those two cases and selecting the optimal solution with lower energy consumption.
	\item A DRL framework for user grouping is designed based on a deep Q-learning algorithm. We provide a training algorithm for the NN to learn the experiences based on the channel condition and delay tolerance of each user during a period of slotted time, and the user grouping policy can be
		learned gradually at the base station by maximizing the negative of the total offloading energy consumption. 
	\item Simulation results are provided to illustrate the convergence speed and the performance of this user grouping policy by comparing with random user grouping policy. Moreover, compared with the OMA-MEC scheme, our proposed NOME-MEC scheme can achieve superior performance with much lower energy consumption.
\end{itemize}

\subsection{Organizations}
The rest of the paper is structured as follows. The system model and the formulated energy minimization problem for our proposed NOMA-MEC scheme are described in Section \ref{sys}. Section \ref{opt}, it presents the optimal solution to the energy minimization problem. Following that, the DRL based user grouping algorithm is introduced in Section \ref{alg}. Finally, the simulation results of the convergence and average performance for the proposed scheme are shown in Section \ref{res}, and Section \ref{con} concludes this paper.

\section{System Model and Problem Formulation} \label{sys}
\subsection{System Model}
In this paper, we consider a NOMA-MEC network, where a base station is equipped with an MEC server to serve $K$ resource-constrained users. During one offloading cycle, each user offloads its task to the MEC server and then obtains the results which processed at the MEC server. Generally, the data size of the computation results is relatively smaller than the offloaded data in practical, thus, the time for downloading the results can be omitted \cite{8537962}. Moreover, since the MEC server has much higher computation capability than mobile devices, the data processing time at the MEC server can be ignored compared to the offloading time \cite{9179779}. Therefore, in this work, the total offloading delay is approximated to the time consumption of data uploading to base station.

We assume that all $K$ users are divided into $\Phi$ groups to transmit signals at different sub-channels, and each group $\phi$ contains two users such that $K=2\Phi$. In each group, we denote the user with short deadline by $U_{m,\phi}$, and the user with relevantly longer deadline by $U_{n,\phi}$, which indicates $\tau_{m,\phi} \leq  \tau_{n,\phi}$, where $\tau_{i,\phi}$ is the latency requirement of $U_{i,\phi},  \forall i \in \{m,n\}$ in group $\phi$. Because $U_{m,\phi}$ has a tighter deadline, it is assumed that the whole duration $\tau_{m,\phi}$ will be used up, which means that the offloading time $t_{m,\phi}=\tau_{m,\phi}$.

In this system model, we adopt the block channel model which indicates that the channel condition remains static during each time slot. With the small scale fading, the channel gain of a user in group $\phi$ can be expressed as

\begin{equation}
	H_{i,\phi} = \tilde{h}_{i,\phi}{{d_{i,\phi}}^{-\frac{\alpha}{2}}}, \quad \forall i \in \{m,n\} ,  \forall \phi,
\end{equation}
where $\tilde{h}_{i,\phi} \sim \mathcal{CN}(0,1) $ is the Rayleigh fading coefficient, $d_{i,\phi}$ is the distance between $U_{i,\phi}$ to the base station, and $\alpha$ is the pass loss exponent. The channel gain is normalized by the addictive white Gaussian noise (AWGN) power with zero-mean and $\sigma^2$ variance, which can be written as 

\begin{equation}
	h_{i,\phi}=\frac{ \abs{H_{i,\phi}} ^2}{\sigma^2}, \quad \forall i \in \{m,n\} ,  \forall \phi.
\end{equation}

As shown in \fref{sysmodel}, since those two users have different delay tolerance, it is natural to consider that the $U_{n,\phi}$ is unnecessary to finish offloading within $\tau_{m,\phi}$ via NOMA transmission, and potentially to save energy if $U_{n,\phi}$ can utilize the spare time $\tau_{n,\phi} - \tau_{m,\phi}$. Hence, our proposed hybrid NOMA scheme enables $U_{n,\phi}$ to offload part of its data when $U_{m,\phi}$ offloading its task during $\tau_{m,\phi}$, an additional time duration $t_{r,\phi}$ is scheduled within each time slot to transmit $U_{n,\phi}$'s remaining data. The task transmission for $U_{m,\phi}$ should be completed within $\tau_{n,\phi}$, i.e.,

\begin{equation}
	t_{r,\phi} \leq \tau_{n,\phi} - \tau_{m,\phi} , \forall \phi.
\end{equation}

As aforementioned, the users in each group will occupy the same sub-channel to upload their data to the base station simultaneously via NOMA. In NOMA uplink transmission, SIC is adopted at the base station to decode the superposed signal. Conventionally, the SIC decoding order is based on either user's CSI or the QoS requirement \cite{9151196}. For the QoS based case, to guarantee $U_{m,\phi}$ can offload its data by $\tau_{m,\phi}$, $U_{n,\phi}$ is set to be decoded first, and the data rate is

\begin{equation} \label{Sys:QoSCon}
	R_{n,\phi} = B\ln \left(1+\frac{P_{n,\phi} \abs{h_{n,\phi}}^2}{P_{m,\phi} \abs{h_{m,\phi}}^2+1}\right),
\end{equation}
where $B$ is the bandwidth of each sub-channel. $P_{n,\phi}$ and $P_{m,\phi}$ are the transmission power of $U_{n,\phi}$ and $U_{m,\phi}$ during NOMA transmission respectively. Based on the NOMA principle, the signal of $U_{m,\phi}$ can then be decoded if \eqref{Sys:QoSCon} is satisfied, and the data rate for $U_{m,\phi}$ can be written as

\begin{equation} 
	R_{m,\phi} = B\ln \left(1+P_{m,\phi}\abs{h_{m,\phi}}^2\right). \label{Sys:OMAUM}
\end{equation}
If $U_{n,\phi}$ is decoded first according to the CSI principle, the achievable rate is same as \eqref{Sys:QoSCon} since $U_{n,\phi}$ treat the signal of $U_{m,\phi}$ as noise power. In contrast, $U_{m,\phi}$ can be decoded first if the following condition holds:

\begin{equation} \label{Sys:QoSCon2}
	R_{m,\phi} \leq B\ln \left(1+\frac{P_{m,\phi}\abs{h_{m,\phi}}^2}{P_{n,\phi}\abs{h_{n,\phi}}^2+1}\right).
\end{equation}
Then the data rate of $U_{n,\phi}$ can be obtained by removing the information of $U_{m,\phi}$, which is

\begin{equation} 
	R_{n,\phi} = B\ln \left(1+P_{n,\phi}\abs{h_{n,\phi}}^2\right). \label{Sys:QoS}
\end{equation}
If the same power is allocated to $U_{n,\phi}$ for both QoS and CSI scheme, it is evident that the achievable rate in \eqref{Sys:QoS} is higher than that in \eqref{Sys:QoSCon}, and the decoding order in \eqref{Sys:QoS} is preferred in  this case. However, since the constraint \eqref{Sys:QoSCon2} cannot be always satisfied, the system has to dynamically change the decoding order accordingly to achieve better performance, which motivated us to utilize the hybrid SIC scheme.

In addition, during $t_{r,\phi}$, $U_{n,\phi}$ adopts OMA transmission, and the data rate can be expressed as
\begin{equation}
	R_{r,\phi} = B\ln \left( 1+ P_{r,\phi}\abs{h_{n,\phi}}^2\right),
\end{equation}
where $P_{r,\phi}$ represents the transmission power of $U_{n,\phi}$ during the second time duration $t_{n,\phi}$.
\begin{figure}[t]
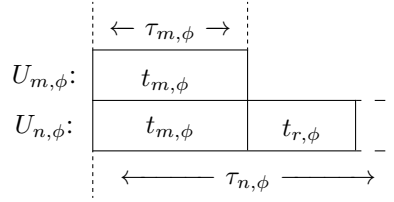

		\centering
		\begin{tabular}{c c c c c c}
			\multicolumn{1}{c}{}&\multicolumn{2}{;{1.5pt/1.5pt} c;{1.5pt/1.5pt}}{\leftarrowfill \  $\tau_{m,\phi}$ \rightarrowfill} \\
			\cline{2-3}
			\multicolumn{1}{c}{$U_{m,\phi}$:} & \multicolumn{2}{|m{1.5cm}|}{\centering $t_{m,\phi}$} & \multicolumn{1}{c}{} 
			& \multicolumn{1}{c}{}\\
			\cline{2-4} \cdashline{5-5}
			\multicolumn{1}{c}{$U_{n,\phi}$:} & \multicolumn{2}{|c|}{$t_{m,\phi}$}&\multicolumn{1}{m{1cm}|}{\centering $t_{r,\phi}$}&\multicolumn{1}{m{0.2cm};{5pt/1pt}}{}\\
			\cline{2-4} \cdashline{5-5}
			\multicolumn{1}{c}{}&\multicolumn{4}{;{1.5pt/1.5pt} c;{1.5pt/1.5pt}}{$\xleftarrow{\quad\quad\quad}$  $\tau_{n,\phi}$  $\xrightarrow{\quad\quad\quad}$}\\	
		\end{tabular}
	\captionsetup{justification=centering,margin=2cm}	
		\caption{System model.}
	\label{sysmodel}
\end{figure}
\par
In this work, the data length of each task is denoted by $L$, which is assumed to be bitwise independent, and we propose a partial offloading scheme in which each task can be processed locally and remotely in parallel. An offloading partition assignment coefficient $\beta_{\phi} \in [ 0,1]$ is introduced, which indicates how much amount of data is offloaded to the MEC server, and the rest can be executed by the local device in parallel. Thus, for each task, the amount of data for offloading to the server is $\beta_{\phi} L$ and $(1-\beta_{\phi})L$ is the data processed locally.
\par
$U_{n,\phi}$ can take the advantage of local computing by executing $(1-\beta_{\phi})L$ data locally during the scheduled NOMA and OMA time duration $t_{m,\phi} + t_{r,\phi} $. Therefore, the energy consumption for $U_{n,\phi}$'s local execution, which is denoted by $E^{loc}_{n,\phi}$, can be expressed as

\begin{equation} \label{local_cmp_E}	
	E^{loc}_{n,\phi} = \frac{\kappa_{0}\left[C(1-\beta_{\phi})L\right]^{3}}{\left(t_{m,\phi}+t_{r,\phi}\right)^{2}},
\end{equation}
where $\kappa_0$ denotes the coefficient related to the mobile device's processor and $C$ is the number of CPU cycles required for computing each bit. 

The total energy consumed by $U_{n,\phi}$ per task involves three parts, including the energy consumed by local computing, and transmission during NOMA and OMA offloading. The power for offloading is scheduled separately during these scheduled two time duration according to the hybrid SIC scheme, and thereby the offloading energy consumption $	E^{off}_{n,\phi}$ can be expressed as

\begin{equation}
	E^{off}_{n,\phi} = t_{m,\phi}P_{n,\phi} + t_{r,\phi}P_{r,\phi}.
\end{equation}
Hence, the total energy consumption can be expressed as

\begin{equation} \label{eng}
	E^{tot}_{\phi}=E^{loc}_{n,\phi}+E^{off}_{n,\phi}.
\end{equation}
\subsection{Problem Formulation}
We assume that the resource allocation of $U_{m,\phi}$ is given as a constant in each group since $U_{m,\phi}$ is treated as the primary user whose requirement need to be guaranteed in priority, and we only focus on the energy minimization for $U_{n,\phi}$ during both NOMA and OMA duration. Given the user grouping policy which will be solved in Section IV, the energy minimization problem for each pair can be formulated as

\begin{subequations}\refstepcounter{prb} \label{P1}
	\begin{align}
		\peqref{P1}: \quad
		\min_{\substack{P_{n,\phi},P_{r,\phi}\\ \mathrm{t_{r,\phi}},\mathrm{\beta_{\phi}}}}& \quad &&
		\frac{\kappa_{0}\left[C(1-\beta_{\phi})L\right]^{3}}{\left(\tau_{m,\phi}+t_{r,\phi}\right)^{2}} 
		+ \tau_{m,\phi}P_{n,\phi} + t_{r,\phi} P_{r,\phi}   \\
		\textrm{s.t.}&\quad &&  
		\tau_{m,\phi} R^{H}_{n,\phi}+t_{r,\phi} B\ln \left( 1+ P_{r,\phi}\abs{h_{n,\phi}}^2\right) \geq \beta_{\phi} L \label{P1C1} \\
		&&& \tau_{m,\phi} B \ln \left(1+ \frac{P_{m,\phi}\abs{h_{m,\phi}}^2}{P_{n,\phi} \abs{h_{n,\phi}}^{2}+1} \right) \geq \onesn L\label{P1C2} \\
		&&& P_{n,\phi} \geq 0, P_{r,\phi} \geq 0 \label{P1C3} \\ 
		&&& 0\leq t_{r,\phi}\leq \tau_{n,\phi}-\tau_{m,\phi} \label{P1C4} \\
		&&& 0\leq \beta_{\phi} \leq 1, \label{P1C5}
	\end{align}
\end{subequations}
where $R^{H}_{n,\phi}=\onesn B\ln\left(1+P_{n,\phi} \abs{h_{n,\phi}}^{2}\right) + (1-\onesn)B\ln\left(1+\frac{P_{n,\phi}\abs{h_{n,\phi}}^2}{P_{m,\phi} \abs{h_{m,\phi}}^2+1}\right) $. $\onesn$ is the indicator function. When $\onesn = 1$, $U_{m,\phi}$ is decoded first and vice verse.  Constraint \eqref{P1C1} and \eqref{P1C2} ensure all the users should complete offloading the designated amount of data within the given deadline. The constraint \eqref{P1C4} limits the additionally scheduled time slot should not beyond $U_{n,\phi}$'s delay tolerance.  Constraints \eqref{P1C3} \eqref{P1C5} set the feasible range of the transmission power and offloading coefficient.
\par
The problem $\peqref{P1}$ is non-convex due to the multiplication of several variables. Therefore, in the following section, we propose a multilevel programming algorithm to address the energy minimization problem optimally by obtaining the closed-form solution.

\section{Energy minimization for NOMA-MEC with Hybrid SIC Scheme} \label{opt}
In this section, a multilevel programming method is introduced to decompose the problem $\peqref{P1}$ into three sub-problems, i.e., power allocation, time slot scheduling and task assignment, which can be solved optimally by obtaining the closed-form solution. The optimal solution to the original problem $\peqref{P1}$ can thereby be found by solving those three sub-problems successively, which are provided in the below subsections.

\subsection{Power Allocation}
Let $t_{r,\phi}$ and $\beta_{\phi} $ be fixed, the problem $\peqref{P1}$ is regarded as a power allocation problem which can be rewritten as

\begin{subequations}\refstepcounter{prb} \label{SP1}
	\begin{align}
		\peqref{SP1}: \quad
		\min_{\substack{P_{n,\phi},P_{r,\phi}}}& \quad &&
		\frac{\kappa_{0}\left[C(1-\beta_{\phi})L\right]^{3}}{\left(\tau_{m,\phi}+t_{r,\phi}\right)^{2}} 
		+ \tau_{m,\phi}P_{n,\phi} + t_{r,\phi} P_{r,\phi}   \\
		\textrm{s.t.}&\quad &&  
		\tau_{m,\phi} R^{H}_{n,\phi}+t_{r,\phi} B\ln \left( 1+ P_{r,\phi}\abs{h_{n,\phi}}^2\right) \geq \beta_{\phi} L \label{P2C1} \\
		&&& \tau_{m,\phi} B \ln \left(1+ \frac{P_{m,\phi}\abs{h_{m,\phi}}^2}{P_{n,\phi} \abs{h_{n,\phi}}^{2}+1} \right) \geq \onesn L\label{P2C2} \\
		&&& P_{n,\phi} \geq 0, P_{r,\phi} \geq 0 \label{P2C3} 
	\end{align}
\end{subequations}
Since there exists an indicator function, $\peqref{SP1}$ is solved in two different cases, i.e., when $\onesn = 1$ and when $\onesn = 0$. The following theorem provides the optimal solution of both cases.

\begin{thm}
	The optimal power allocation to $\peqref{SP1}$ is given by the following two cases according to the indicator function:
	\begin{enumerate}
		\item For $\onesn = 1$, $U_{m,\phi}$ is decoded first, and the power allocation for this decoding order is presented as follows:
		\begin{enumerate}
		
			\item	When $P_{n,\phi}\neq 0$ and $P_{r,\phi}\neq 0$, $U_{n,\phi}$ offloads in both time duration, which is termed as hybrid NOMA, and the power allocation is given in the following two cases:
				\begin{enumerate}
					\item If $P_{m,\phi} > \abs{h_{m,\phi}}^{-2}e^{\frac{\beta_{\phi} L }{B \left( \tau_{m,\phi}+t_{r,\phi} \right) } }\left(e^{\frac{L}{B\tau_{m,\phi} }} -1\right)$,
						\begin{equation} \label{OptP1}
							P_{n,\phi}^{*}=P_{r,\phi}^{*}=\abs{h_{n,\phi}}^{-2} \left( e^{\frac{\beta_{\phi} L }{B \left( \tau_{m,\phi}+t_{r,\phi} \right) } } -1\right).
						\end{equation}
					\item If $ \abs{h_{m,\phi}}^{-2}\left(e^{\frac{L}{B\tau_{m,\phi} }} -1\right)    \leq P_{m,\phi} \leq \abs{h_{m,\phi}}^{-2}e^{\frac{\beta_{\phi}L}{B\tau_{m,\phi}}}\left(e^{\frac{L}{B\tau_{m,\phi}}}-1 \right)$,
						\begin{subequations}\label{OptP2}				
					\begin{align}				
						&	P_{n,\phi}^{*} = \abs{h_{n,\phi}}^{-2}\left[P_{m,\phi}\abs{h_{m,\phi}}^{2} \left(e^{\frac{L}{B\tau_{m,\phi}}}-1\right)^{-1}-1  \right],\\
						& 	P_{r,\phi}^{*} = \abs{h_{n,\phi}}^{-2} \left\{ e^{\frac{\beta_{\phi} L}{B t_{r,\phi}} - \frac{\tau_{m,\phi}}{t_{r,\phi}} \ln \left[P_{m,\phi}\abs{h_{m,\phi}}^{2}\left( e^{\frac{L}{B\tau_{m,\phi}}}-1 \right)^{-1} \right]  } -1 \right\}.
					\end{align}
				\end{subequations}
				\end{enumerate}

			\item When $U_{n,\phi}$ only offloads during the first time duration $\tau_{m,\phi}$, this scheme is termed as pure NOMA, and the power allocation is obtained as
			
			\rm{if} $P_{m,\phi} \geq \abs{h_{m,\phi}}^{-2}e^{\frac{\beta_{\phi}L}{B\tau_{m,\phi}}}\left(e^{\frac{L}{B\tau_{m,\phi}}}-1 \right)$,
			
			\begin{subequations}
				\begin{align}
					& P_{n,\phi}^{*} =\abs{h_{n,\phi}}^{-2}\left( e^{\frac{\beta_{\phi}L}{B\tau_{m,\phi}}}-1\right),\\
					& P_{r,\phi}^{*} =0.
				\end{align}
			\end{subequations}

			\item When $P_{n,\phi}^{*}=0$, $U_{n,\phi}$ chooses to offload solely during the section time duration $t_{r,\phi}$, and the optimal power allocation is:
			
			\rm{if} $P_{m,\phi}\geq \abs{h_{m,\phi}}^{-2} \left(e^{\frac{L}{B\tau_{m,\phi} }} -1\right) $,
			
			\begin{subequations}
				\begin{align}
					& P_{n,\phi}^{*}=0,\\
					& P_{r,\phi}^{*} = \abs{h_{n,\phi}}^{-2}\left( e^{\frac{\beta_{\phi} L }{Bt_{r,\phi}}}-1 \right).\label{pure1}
				\end{align}
			\end{subequations}
			
		\end{enumerate}
		\item For $\onesn = 0$:
		\begin{itemize}
			\item[1)]	When $P_{n,\phi}\neq 0$ and $P_{r,\phi}\neq 0$, $U_{n,\phi}$, the hybrid NOMA power allocation is given by
			
			if $P_{m,\phi} \leq |h_{m,\phi}|^{-2} \left(e^{\frac{\beta_{\phi}L}{t_{r,\phi}}}-1\right)$,
						
			\begin{subequations}\label{OptP3}
				\begin{align}
					P_{n,\phi}^{*}=& \abs{h_{n,\phi}}^{-2} \left(P_{m,\phi} \abs{h_{m,\phi}}^{2} +1\right) \left[e^{\frac{\beta_{\phi}L-t_{r,\phi} \ln\left(P_{m,\phi} \abs{h_{m,\phi}}^{2} +1\right)}{B\left(\tau_{m,\phi}+t_{r,\phi}\right)}}-1\right]
					\\ 
					P_{r,\phi}^{*}=&\abs{h_{n,\phi}}^{-2} \left[\left(P_{m,\phi} \abs{h_{m,\phi}}^{2} +1\right)e^{\frac{\beta_{\phi}L-t_{r,\phi} \ln\left(P_{m,\phi} \abs{h_{m,\phi}}^{2} +1\right)}{B\left(\tau_{m,\phi}+t_{r,\phi}\right)}}-1\right].
				\end{align}
			\end{subequations}

			\item[2)] When $P_{r,\phi} = 0$, the pure NOMA case can be obtained as
			\begin{equation}
				P_{n,\phi}^{*}= \abs{h_{n,\phi}}^{-2}\left(P_{m,\phi} \abs{h_{m,\phi}}^{2} +1\right)\left(e^{\frac{\beta_{\phi}L}{B\tau_{m,\phi}}}-1\right).
			\end{equation}
			\item[3)] When $P_{n,\phi} = 0$, the OMA case is:
			\begin{subequations}
				\begin{align}
					& P_{n,\phi}=0,\\
					& P_{r,\phi} = \abs{h_{n,\phi}}^{-2}\left( e^{\frac{\beta_{\phi} L }{Bt_{r,\phi}}}-1 \right).
				\end{align}
			\end{subequations}
			
		\end{itemize}
	\end{enumerate}
\end{thm}
\begin{proof}
	Refer to Appendix A.
\end{proof}
\begin{remark}
	Theorem 1 provides the optimal power allocation for both two decoding sequences, i.e., $U_{m,\phi}$ is decode first when $\onesn = 1$, and $U_{n,\phi}$ is decode first when $\onesn = 0$. The optimal solution to  $\peqref{P1}$ is obtained by numerical comparison between these two cases in terms of energy consumption. Both cases can be further divided into three offloading scenarios including hybrid NOMA, pure NOMA and OMA based on different power allocation. For hybrid NOMA case, $U_{n,\phi}$ transmits during both $\tau_{m,\phi}$ and $t_{r,\phi}$, which indicates $P_{n,\phi}>0$, $P_{r,\phi}>0$ and $t_{r,\phi}>0$. Pure NOMA scheme indicates that $U_{n,\phi}$ only transmits simultaneously with $U_{m,\phi}$ during $\tau_{m,\phi}$, and therefore, $P_{r,\phi}=0$ and $t_{r,\phi}=0$. In addition, the OMA case represents that $U_{m,\phi}$ occupies $\tau_{m,\phi}$ solely, and $U_{m,\phi}$ only transmit during $t_{r,\phi}$.
\end{remark}
\begin{remark}
	Appendix A provides the proof for the case $\onesn = 1$. The proof for the case $\onesn = 0$ similarly, and it can be referred to the previous work in \cite{LI2020241}. Thus, the proof for the case $\onesn = 0$ is omitted for this and the following two sub-problems. 
\end{remark}
In this subsection, the optimal power allocation for the hybrid NOMA scheme is obtained when $t_{r,\phi}$ is fixed, and then the optimization of $t_{r,\phi}$ is further studied to minimize $E^{tot}_{n,\phi}$ in the following subsection.
\subsection{Time Schedualing}
The aim of this subsection is to find the optimal time allocation for the second time duration $t_{r,\phi}$ which is solely utilized by $U_{n,\phi}$ for OMA transmission. As aforementioned in Theorem 1, the optimal power allocation for hybrid NOMA scheme is given as a function of $t_{r,\phi}$ and $\beta_{\phi}$. Hence, by fixing $\beta_{\phi}$, $\peqref{P1}$ is rewritten as

\begin{subequations}\refstepcounter{prb} \label{P1S2}
	\begin{align}
		\peqref{P1S2}:\quad
		\min_{\mathrm{t_{r,\phi}}}& \quad &&
		\frac{\kappa_{0}\left[C(1-\beta_{\phi})L\right]^{3}}{\left(\tau_{m,\phi}+t_{r,\phi}\right)^{2}} 
		+ \tau_{m,\phi}P^{*}_{n,\phi} + t_{r,\phi} P^{*}_{r,\phi}\label{P1S2O1}   \\
		\textrm{s.t.}
		&&& 0\leq t_{r,\phi}\leq \tau_{n,\phi}-\tau_{m,\phi} \label{P1SC4}
	\end{align}
\end{subequations}
\begin{prop}
	The offloading energy consumption $\eqref{P1S2O1}$ is monotonically decreasing with respected to $t_{r,\phi}$ for both $\onesn = 1$ and $\onesn = 0$ cases. To minimize the energy consumption, the optimal time allocation is to schedule the entire available time before the deadline $\tau_{n,\phi}$, i.e., 
	
	\begin{equation}
		t^{*}_{r,\phi} = \tau_{n,\phi} - \tau_{m,\phi}
	\end{equation}
\end{prop}
\begin{proof}
	Refer to Appendix B.
\end{proof}
By assuming all the data is offloaded to the MEC server, the following lemma studies the uplink transmission energy efficiency of the two hybrid NOMA-MEC schemes for $\onesn=0$ and $\onesn=1$.
\begin{lem}
	Assume all data are offloaded to the MEC server, i.e., $\beta_{\phi}=1$, the solution in \eqref{OptP3} for the case $\onesn=0$ has higher energy consumption than the solution in \eqref{OptP1} for the case $\onesn=1$, if $ \abs{h_{m,\phi}}^{-2}\left(e^{\frac{L}{B\tau_{m,\phi} }} -1\right) \leq  P_{m,\phi} \leq |h_{m,\phi}|^{-2} \left(e^{\frac{L}{\tau_{n,\phi}-\tau_{m,\phi}}}-1\right)$.
\end{lem}
\begin{proof}
	Without considering local computing, the energy consumption for \eqref{OptP1} can be written as
	\begin{equation}
		E_{1} = \tau_{n,\phi}\abs{h_{n,\phi}}^{-2} \left( e^{\frac{L }{B \tau_{n,\phi} } } -1\right),
	\end{equation}
	and the energy consumption for the case \eqref{OptP3} is given as
	\begin{equation}
		\begin{split}
			E_{2} =& \tau_{m,\phi}\abs{h_{n,\phi}}^{-2} \left(P_{m,\phi} \abs{h_{m,\phi}}^{2} +1\right) \left[e^{\frac{L- \left(\tau_{n,\phi}-\tau_{m,\phi}\right) \ln\left(P_{m,\phi} \abs{h_{m,\phi}}^{2} +1\right)}{B\tau_{n,\phi}}}-1\right] \\&+ \left(\tau_{n,\phi}-\tau_{m,\phi}\right) \abs{h_{n,\phi}}^{-2} \left[\left(P_{m,\phi} \abs{h_{m,\phi}}^{2} +1\right)e^{\frac{L-\left(\tau_{n,\phi}-\tau_{m,\phi}\right)\ln\left(P_{m,\phi} \abs{h_{m,\phi}}^{2} +1\right)}{B\tau_{n,\phi}}}-1\right].
		\end{split}		
	\end{equation}
	To proof that $E_{2} \geq E_{1}$, the inequality can be rearranged as	
	\begin{equation}
		-\tau_{m,\phi} P_{m,\phi}\abs{h_{m,\phi}}^{2} + 
		\tau_{n,\phi}e^{\frac{L}{B\tau_{n,\phi}}}\left(P_{m,\phi} \abs{h_{m,\phi}}^{2} +1 \right)^{\frac{\tau_{m,\phi}-\tau_{n,\phi}}{B\tau_{n,\phi}}+1}\geq \tau_{n,\phi} e^{\frac{L}{B\tau_{n,\phi}}}.
	\end{equation}
	Define $ \zeta(x) = -\tau_{m,\phi}x + \tau_{n,\phi}e^{\frac{L}{B\tau_{n,\phi}}}(x+1)^{\frac{\tau_{m,\phi}-\tau_{n,\phi}}{B\tau_{n,\phi}}+1} $,the first order derivative of $ \zeta(x)$ is given as	
	\begin{equation}
		\zeta^{'}(x) = -\tau_{m,\phi} + (\frac{\tau_{m,\phi}-\tau_{n,\phi}}{B\tau_{n,\phi}}+1)\tau_{n,\phi}e^{\frac{L}{B\tau_{n,\phi}}}(x+1)^{\frac{\tau_{m,\phi}-\tau_{n,\phi}}{B\tau_{n,\phi}}}.
	\end{equation}
	Therefore, $\zeta^{'}(x)$ is monotonically decreasing since $\tau_{m,\phi}<\tau_{n,\phi}$, and the following inequality holds:
	\begin{equation}
		\zeta^{'}(x) \geq \zeta^{'}\left(e^{\frac{L}{\tau_{n,\phi}-\tau_{m,\phi}}}-1\right)=0.
	\end{equation}
	Hence for $0 \leq x \leq e^{\frac{L}{\tau_{n,\phi}-\tau_{m,\phi}}}-1$,  $\zeta(x)$ is monotonically increasing, and $\zeta(x) \geq \zeta(0)= \tau_{n,\phi} e^{\frac{L}{B\tau_{n,\phi}}}$, which illustrates that $E_{2} \geq E_{1}$.
\end{proof}

\subsection{Offloading Task Assignment}
In this subsection, we focus on the optimization of the task assignment coefficient for $U_{n.\phi}$ in each group $\phi$. Given the optimal power allocation and time arrangement,  $\peqref{P1}$ is reformulated as

\begin{subequations} \refstepcounter{prb}\label{Pn4r}
	\begin{align}
		\peqref{Pn4r}:\quad
		\min_{\beta_{\phi}}&  &&
		\frac{\kappa_{0}\left[C(1-\beta_{\phi})L\right]^{3}}{\left(\tau_{m,\phi}+t^{*}_{r,\phi}\right)^{2}} 
		+ \tau_{m,\phi}P^{*}_{n,\phi} + t^{*}_{r,\phi} P^{*}_{r,\phi}   \\
		\textrm{s.t.}&\quad &&   0\leq \beta_{\phi} \leq 1, \label{P4C}
	\end{align}
\end{subequations}
\begin{prop}
	The above problem is convex, and the optimal task assignment coefficient can be characterized by those three optimal power allocation schemes for the hybrid NOMA model in \eqref{OptP1}, \eqref{OptP2}, and \eqref{OptP3}, which is given by
	\begin{equation}
		\beta_{\phi}^{*}=1-\frac{2}{z_{2,\phi}}\mathcal{W}\left(\frac{1}{2}z_{1,\phi}^{-\frac{1}{2}}z_{2,\phi}e^{\frac{z_{2,\phi}}{2}}\right),
	\end{equation}
		where $\mathcal{W}$ denotes the single-valued Lambert W function, and $z_{1,\phi}$ and $z_{2,\phi}$ are determined by the different power allocation schemes, which are presented as follows:
	\begin{itemize}
		\item[(a)] $\onesn = 1$:	
		
		If \eqref{OptP1} is adopted:
		
		\begin{equation}
			\begin{cases}
				&z_1 = \frac{3\kappa_{0}BC^{3}L^{2}\abs{h_{n,\phi}}^{2}}{\tau_{n,\phi}^{2}}, \quad  \\
				&z_2 =  \frac{L}{B\tau_{n,\phi}}
			\end{cases}
		\end{equation}
	
		If \eqref{OptP2} is adopted:
		
		\begin{equation}
			\begin{cases}
				&z_1 = \frac{3\kappa_{0}B\abs{h_{n,\phi}}^{2}C^{3}L^{2}e^{2u_{\phi}} }{\tau_{n,\phi}^{2}}  \\
				&z_2 =  \frac{L}{B\left(\tau_{n,\phi}-\tau_{m,\phi}\right)}
			\end{cases}
		\end{equation}
		where $u_{\phi}= \frac{\tau_{m,\phi}}{\left(\tau_{n,\phi}-\tau_{m,\phi}\right)} \ln \left[P_{m,\phi}\abs{h_{m,\phi}}^{2}\left( e^{\frac{L}{B\tau_{m,\phi}}}-1 \right)^{-1} \right]$
		
		\item[(b)] $\onesn = 0$:\\
		\begin{equation}
			\begin{cases}
				&z_{1,\phi} =\frac{3 \kappa_{0} B C^{3}L^{2}  |h_{n,\phi}|^{2} e^{\frac{(\tau_{n,\phi}-\tau_{m,\phi})\ln \left(P_{m,\phi} \abs{h_{m,\phi}}^{2} +1\right)}{B\tau_{n,\phi}}}}{{\tau_{n,\phi}}^2 \left(P_{m,\phi} \abs{h_{m,\phi}}^{2} +1\right)}  \\
				&z_{2,\phi} =  \frac{L}{B\tau_{n,\phi}}
			\end{cases}
		\end{equation}
		
	\end{itemize}	
\end{prop}
\begin{proof}
	Refer to Appendix C
\end{proof}
\begin{remark}
	Problem $\peqref{Pn4r}$ is the lowest level of the proposed multilevel programming method, which provides three task assignment solutions corresponding to the three power allocation schemes \eqref{OptP1}, \eqref{OptP2}, and \eqref{OptP3} respectively. The final solution to the energy minimization problem $\peqref{P1}$ can be obtained by substituting the optimal task assignment into the corresponded power allocation schemes. Then the most energy efficient scheme is selected among  \eqref{OptP1}, \eqref{OptP2}, and \eqref{OptP3} by comparing the numerical energy consumption for each scheme. 
\end{remark}


\section{Deep Reinforcement Learning Framework for User Grouping} \label{alg}
In the previous section, it is assumed that the user grouping is given, and the optimal resource allocation is obtained in closed-form. The optimal user grouping can be obtained by exploring all possible user grouping combinations and find the one with the lowest energy consumption. Although this method can obtain the optimal user pairing scheme, the complexity of the exhaustive search method is high, and it is not possible to output real time decisions. Therefore, we propose a fast converge user pairing training algorithm based on DQN to obtain the user grouping policy, which is introduced in the following subsection, in which the state space, action space and reward function are defined. Subsequently, the training algorithm for the user grouping policy is provided.
\subsection{The DRL Framework}
The optimization of user grouping is modeled as a DRL task, where the base station is treated as the agent to interact with the environment which is defined as the MEC network. In each time slot $t$, the agent takes an action $a_{t}$ from the action space $\mathcal{A}$ to assign users into pairs according to an optimal policy which is learned by the DNN. The action taken under current state $s_{t}$ results an immediate reward $r_{t}$, which is obtained at the beginning of the next time slot, and then move to the next state $s_{t+1}$. In this problem, the aforementioned terms are defined as follows.
\begin{enumerate}[label=\arabic*)]
	\item \emph{State Space:} The state $s_{t} \in \mathcal{S}$ is characterized by the current channel gains and offloading deadlines of all users since the user grouping is mainly determined by those two factors. Therefore, the state  $s_{t}$ can be expressed as
	
	\begin{equation}
		s_{t} = \{h_{1}[t],h_{2}[t],...,h_{k}[t],...,h_{K}[t]; \tau_{1}[t],\tau_{2}[t],...,\tau_{k}[t],...,\tau_{K}[t] \}.
	\end{equation}
	\item \emph{Action Space:} At each time slot $t$, the agent takes a action $a_{t} \in \mathcal{A}$, which contains all the possible user grouping decisions $j_{k,\phi}$. The action is defined as
	\begin{equation}
		a_{t} = \{j_{1,1}[t],...j_{k,\phi}[t],...j_{K,\Phi}[t] \},
	\end{equation}
	where $j_{k,\phi} = 1$ indicates that $U_{k}$ is assigned to group $\phi$. In our proposed scheme, each group can only be assigned with two different users.
	\item \emph{Rewards:} The immediate reward $r_{t}$ is described by the sum of the energy consumption of each groups after choosing the action $a_{t}$ under state $s_{t}$. The numerical result of the energy consumption in each group can be obtained by solving the problem $\peqref{P1}$. Therefore, the reward is defined as
	
	\begin{equation} \label{rwd}
		r_{t} =-\sum_{\phi=1}^{\Phi} E^{tot}_{\phi}[t]
	\end{equation}
\end{enumerate}
The aim of the agent is to find an optimal policy that maximizes the long-term discounted reward, which can be written as

\begin{equation}
	\begin{split}
		R_{t} &= r_{t} + \gamma r_{t+1} + \gamma^2 r_{t+2} +...\\
		&=\sum_{i=0}^{\infty}\gamma^{i} r_{t+i},
	\end{split}	
\end{equation}
where $\gamma \in [0,1]$ is the discount factor which balance the immediate reward and the long-term reward.
\subsection{DQN-based NOMA User Grouping Algorithm}
To accommodate the reward maximization problem, a DQN-based user grouping algorithm is proposed in this paper, illustrated in \fref{nnmodel}. In the conventional Q-learning, Q-table is obtained to describe the quality of an action for a given state, and the agent chooses actions according to the Q-values to maximize the reward. However, it will be slow for the system to obtain Q-values for all the state-action pairs if the state space and action space are large. Therefore, to speed up the learning process, instead of generating and processing all possible Q-values, DNNs are introduced to estimate the Q-values based on the weight of DNNs. We utilize a DNN to estimate the Q-value denoted by Q-network, which the Q-estimation is represented as $Q(s_t,a_t;\theta)$, and an additional DNN with the same setting to generate the target network with $Q(s_t, a_t;\theta^-)$ for training, where $\theta$ and $\theta^-$ are the weights of the DNNs. 

 We adopt $\epsilon$-greedy policy  with $0<\epsilon<1$ to balance the exploration of new actions and the exploitation of known actions by either randomly choosing an action $a_t \in \mathcal{A}$ with probability $\epsilon$ to avoid the agent sticking on non-optimal actions or picking the best action with the probability $1-\epsilon$ such that \cite{MnihKSRVBGRFOPB15}:
 
 \begin{equation}
 	a_t = \arg\max_{a_t \in \mathcal{A}}Q(s_t, a_t; \theta).
 \end{equation}
  Generally, the threshold $\epsilon$ is fixed, which indicates the probability of choosing random action remains the same throughout the whole learning period. However, it brings fluctuation when the algorithm converges and may lead to diverge again in extreme cases. In this paper, we adopt an $\epsilon$-greedy decay scheme, which a large $\epsilon^{+}$ (more greedy) is given at the beginning, and then the it decays with each training step until a certain small probability $\epsilon^{-}$. The above policy encourages the agent to explore the never-selected actions at the beginning, and then the agent intends to take more large reward-guaranteed actions when the network is already converged.
\begin{figure}
	\centering
	\includegraphics[scale=0.5, trim=4.5cm 10 0 10, clip]{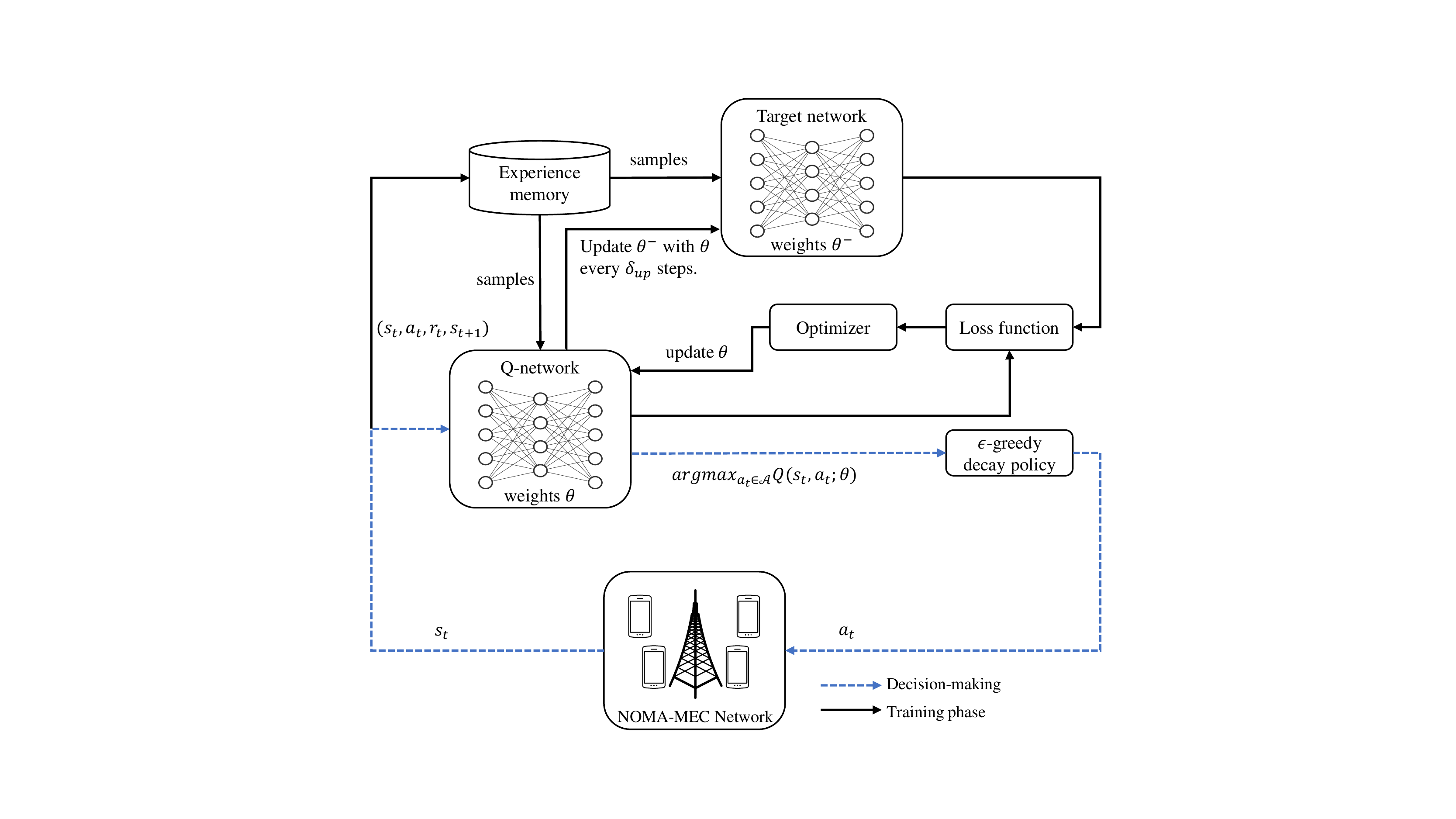}
	\caption{A demonstration of the proposed DQN-based user grouping scheme in the NOMA-MEC network.}
	\label{nnmodel}
\end{figure}

\begin{algorithm}[t]
	\caption{DQN-based User Grouping Algorithm}
	\label{alg:alg1}
	\begin{algorithmic}[1]
		\STATE \textbf{Parameter initialization:}
		\STATE Initialize Q-network $Q(s_{i},a_{i};\theta) $ and target network $Q(s_{i},a_{i};\theta^{-})$.
		\STATE Initialize Reply memory $\mathcal{R}$ with size $|\mathcal{R}|$, and memory counter.
		\STATE Initialize $\gamma$, $\epsilon^{+}$, $\epsilon^{-}$, decay step, batch size, target network update interval $\delta_{up}$.
		\STATE \textbf{Training Phase:}
		\FOR{$episode = 1, 2,..., N_{episode}$}
		\FOR{$time\ step = 1, 2,..., N_{ts}$}
		\STATE Input state $s_t$ into Q-network and obtain Q-values for all actions.
		\STATE Take the user grouping decision as action $a_t$ based on the $\epsilon$-greedy decay policy.
		\STATE Agent receive the reward $r_t$ based on \eqref{rwd} and the observation to next state $s_{t+1}$.
		\STATE Store the experience tuple $(s_t, a_t, r_t, s_{t+1})$ into the memory $\mathcal{R}$.
		\IF{memory counter $>|\mathcal{R}|$} 
		\STATE {Remove the old experiences from the beginning.}
		\ENDIF
		\STATE Randomly sample a mini-batch of the experience tuples $(s_t, a_t, r_t, s_{t+1})$ with batch size and feed into the DNNs.
		\STATE Update the Q-network weights $\theta$ by calculating the Loss function $\eqref{loss}$.
		\STATE Replace $\theta^{-}$ by $\theta$ after every $\delta_{up}$ steps.
		\ENDFOR 
		\ENDFOR		
	\end{algorithmic}
\end{algorithm}

The target network only updates every certain iterations, which provides a relatively stable label for the estimation network. The agent stores the tuples $(s_t, a_t, r_t, s_{t+1})$ as experiences to a memory buffer $\mathcal{R}$, and a mini-batch of samples from the memory are fed into the target network to generate the Q-values labels, which is given by

\begin{equation}
	y_i = r_i + \max_{a_{i+1}\in\mathcal{A}} Q(s_{i+1},a_{i+1};\theta^{-}), \quad \forall i \in \mathcal{R}
\end{equation}
Hence, the loss function for the Q-network can be expressed as

\begin{equation}\label{loss}
	Loss(\theta) = \left(y_{i} - Q(s_{i},a_{i};\theta) \right) , \quad \forall i \in \mathcal{R}
\end{equation}
The Q-network can be trained by minimizing the loss function to obtain the new $\theta$, and the weights of the target network is updated after $\delta_{up}$ steps by replacing $\theta^{-}$ with $\theta$. The whole DQN-based user grouping framework is summarized in Algorithm \ref{alg:alg1}.

\section{Simulation Results} \label{res}
In this section, several simulation results are presented to evaluate the convergence and effectiveness of the proposed joint resource allocation and user grouping scheme. Specifically, the impact of learning rate, user number, offloading data length, and delay tolerance are investigated. Moreover, the proposed hybrid SIC scheme is compared to some benchmarks including QoS based SIC scheme and other NOMA and OMA schemes.

The system parameters are set up as follows. All users are distributed uniformly and randomly in a disc-shape cell where the base station located in the cell center. The total number of users is six, and each of them has a task contains $2$ Mbit of data for offloading. As aforementioned, the delay sensitive primary user $U_{m,\phi}$ is allocated with a predefined power which is $P_{m,\phi} = 1$ W for all groups in the simulation. The delay tolerance for each user is given randomly between $[0.2,0.3]$ seconds. In addition, the rest of the system parameters are listed in Table \ref{Para_sys}.
\begin{table}
	\small
	\caption{System parameters} 
	\label{Para_sys}
	\centering
	\begin{tabular}{|l|c|}   	
		\hline
		Effective capacitance coefficient & $10^{-28} $
		\\
		\hline
		Number of CPU cycles required per bit & $10^3$\\
		\hline	
		Transmission bandwidth $B$ & $2$ MHz \\
		\hline
		Path loss exponent $\alpha$ & $3.76$ \\
		\hline	
		Noise spectral density $N_0$ & $-174$ dBm/Hz	\\
		\hline
		Maximum cell radius & $1000$ m\\
		\hline
		Minimum distance to base station & $50$ m\\
		\hline
	\end{tabular}   	    	  	
\end{table}
\begin{table}
	\small
	\caption{Hyper-parameters} 
	\label{Para_hp}
	\centering
	\begin{tabular}{|l||c|}   	
		\hline
		$\epsilon$-greedy coefficient  & $0.5-0.01$
		\\
		\hline
		$\epsilon$-greedy decay steps  & $2000$ \\
		\hline
		Discount factor $\gamma$  & $0.7$ \\
		\hline	
		Reply memory size $\mathcal{R} $ & $20000$	\\
		\hline
		Batch size & $64 $\\
		\hline
		Target network update interval $\delta_{up}$ & $10$\\
		\hline
		Number of episode $N_{episode}$ & $150$\\
		\hline
		Number of time steps $N_{ts}$ & $500$\\
		\hline
	\end{tabular}   	    	  	
\end{table}
\par
To implement the DQN algorithm, the two DNNs are configured with the same settings, where each of them consists of four fully connected layers, and two of which are hidden layers with 200 and 100 neurons respectively. The activation function we adopted for all hidden layers is Rectified Linear Unit (ReLU), i.e., $f(x)=\max(0,x)$, and the final output layer is activated by Tanh of which the range is $(-1, 1)$ \cite{relu2018}. The Adaptive moment estimation optimizer (Adam) method is used to learn the DNN weight $\theta$ with given learning rate \cite{Adam2014}. The rest of the hyper-parameters are listed in Table \ref{Para_hp}. All simulation results are obtained with PyTorch 1.70 and CUDA 11.1 on Python 3.8 platform.

\subsection{Convergence of Framework}
\begin{figure} 	
	\centering
	\includegraphics[scale=0.6]{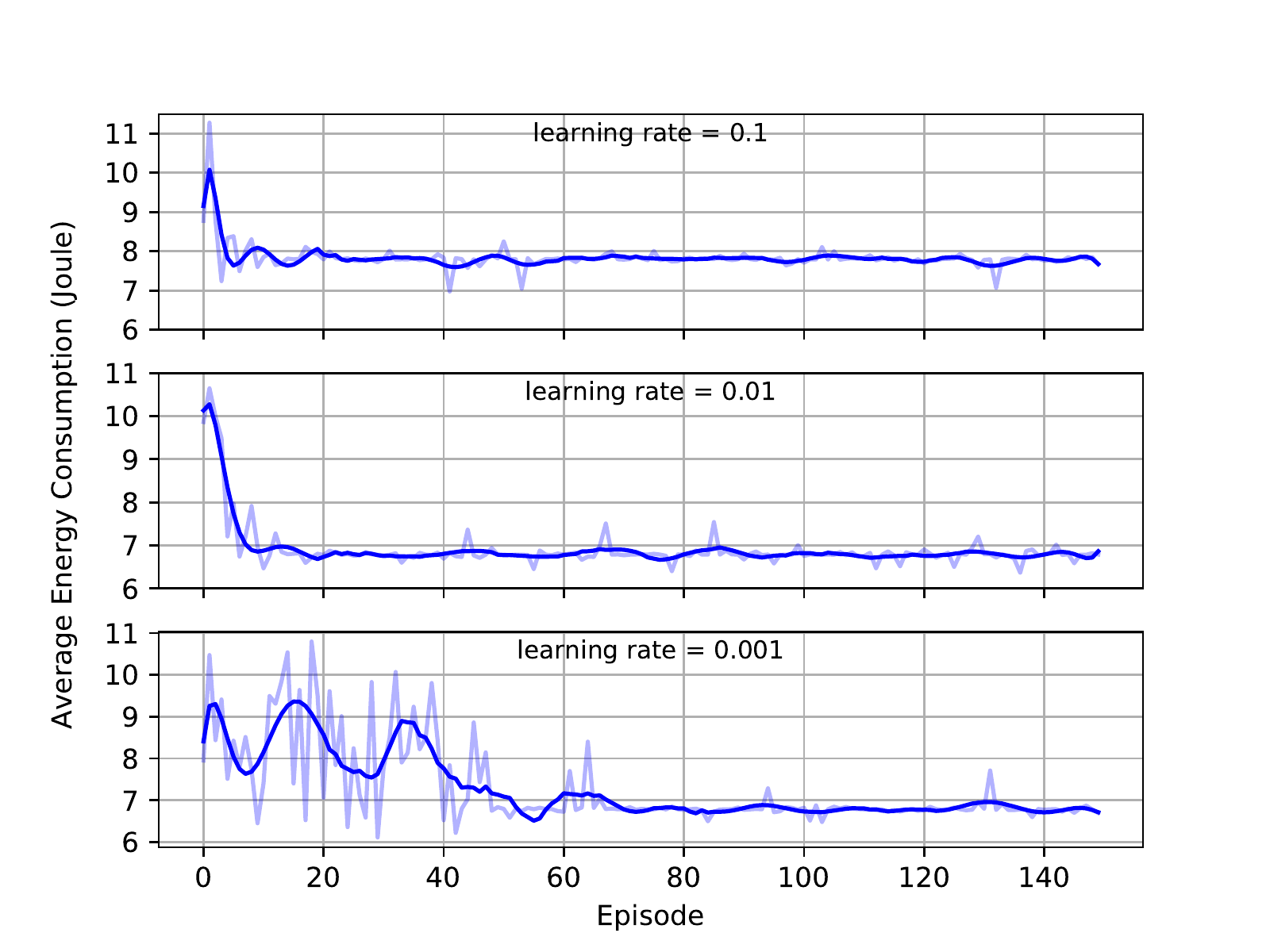}
	\caption{Average energy consumption versus training episodes with different learning rate.}
	\label{r1}
\end{figure}
\begin{figure}
	\centering
	\includegraphics[scale=0.6]{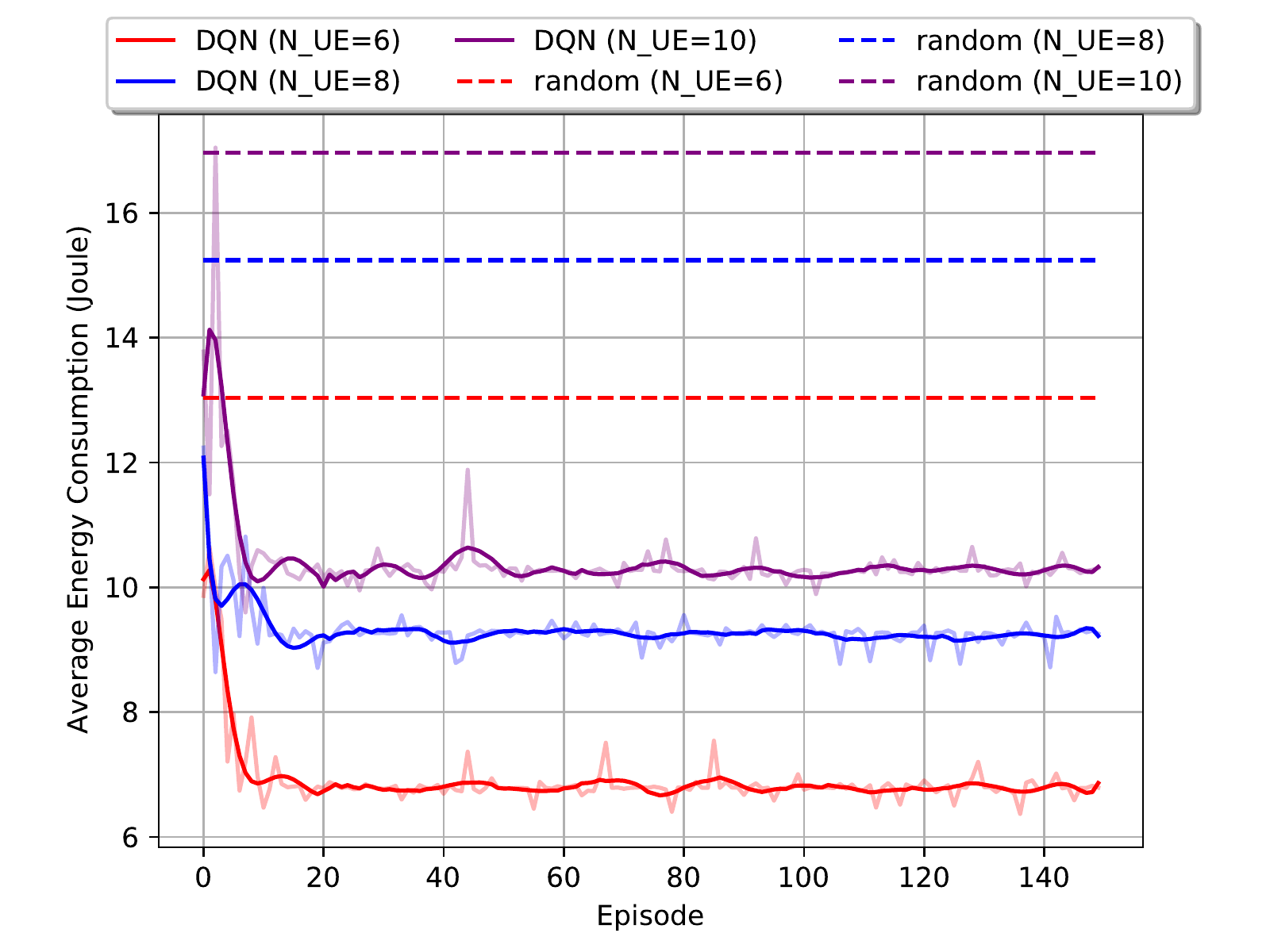}
	\caption{Average energy consumption versus training episodes with different numbers of users. }
	\label{r2}
\end{figure}
\begin{figure}
	\centering
	\includegraphics[scale=0.6]{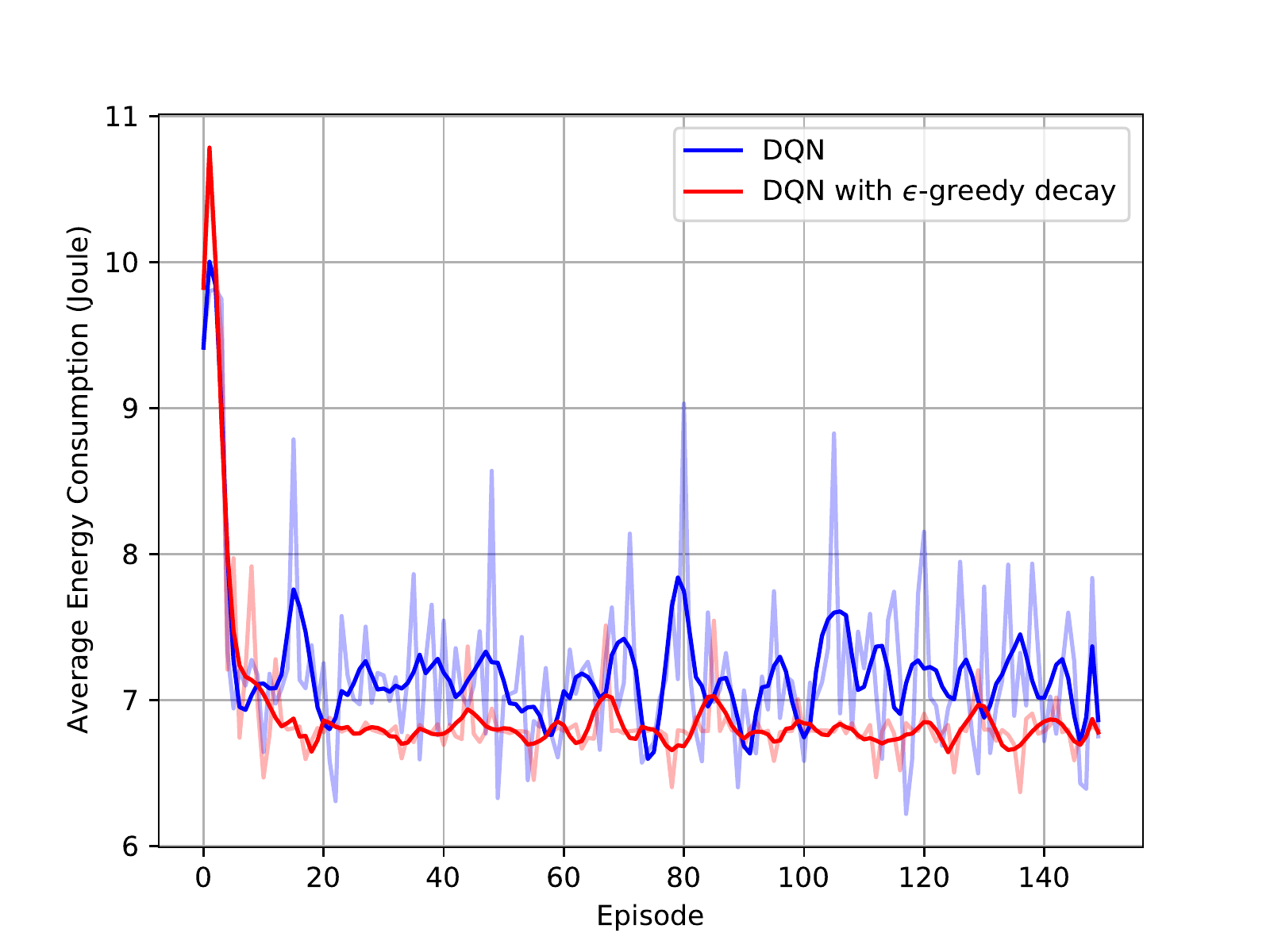}
	\caption{Average energy consumption versus training episodes with different numbers of users. }
	\label{r3}
\end{figure}
In this part, we evaluate the convergence of the proposed DQN based user pairing algorithm. \fref{r1} compares the convergence rate of the average reward for each episode under different learning rate, which is described by the average energy consumption. Learning rate controls how much it should be to adjust the weights of a DNN based on the network loss, and we set the learning rate $=[0.1, 0.01,0.001]$ to observe its influence to the convergence. The network with $0.1$ learning rate converges slightly faster than the one with $0.01$ learning rate, and both of them converge much faster than the network with $0.001$ learning rate. However, when the learning rate is $0.1$, even though the large learning has better convergence, it overshoots the minimum and therefore has higher energy consumption after converge than other two plots. Therefore, the most suitable learning rate for our proposed DQN algorithm is $0.01$, which is adopted to obtain the rest of simulation results in this paper.
\par 
\fref{r2} illustrates the effectiveness of the DQN user grouping algorithm proposed in this paper. By setting the numbers of users to $[6, 8, 10]$, the algorithm shows a similar performance that the average energy consumption decreases over training and converges within the first 20 episodes for the all three cases. Moreover, more users in the network can result in higher energy consumption, and the algorithm shows the superior performance over the random policy, which reduces the energy consumption significantly.
\par
The $\epsilon$-greedy decay policy to the convergence performance is further investigated in \fref{r3}. The $\epsilon$-greedy coefficient for the blue curve is set to $0.1$ while the red curve adopts the $\epsilon$-greedy decay policy following the parameters in Table \ref{Para_hp}. Since the decay policy starts with large $\epsilon$, the network is more likely to choose the random action at the beginning, and hence the energy consumption is higher at the beginning. With $\epsilon$ decays over episode, the network chooses the actions which have been selected before that guarantees large rewards, and therefore it is more stable afterwards. Meanwhile, the network without the decay policy has significant fluctuations during the training since it has a  greater chance to choose the random actions throughout the training. However, if a very small $\epsilon$ is adopted, the network will be less likely to explore some actions, which may stick on the non-optimal actions.

\subsection{Average Performance of Proposed Scheme}

In this part, we present the average performance of the proposed NOMA-MEC scheme to show the impact of $P_{m,\phi}$, offloading data length, and maximum delay tolerance. Meanwhile, our proposed scheme is compared with the one without task assignment and OMA offloading to show the superior performance gap. As shown in \fref{r4}, the energy consumption of both hybrid-SIC schemes raises and then decreases as $P_{m,\phi}$ increases. Since $P_{m,\phi}$ is relatively small at the beginning, $U_{m,\phi}$ is not likely to be decoded first to satisfy the constraint \eqref{P1C2} in the case $\onesn=1$. Therefore, $U_{n,\phi}$ is morel likely to be decoded in priory, and increasing $P_{m,\phi}$ causes more interference to $U_{n,\phi}$. After the power indicated by the arrows, the case $\onesn=1$ becomes feasible for both with and without task assignment schemes, and it is evident that the case $\onesn=1$ has better energy efficiency compared to the case $\onesn=0$. Moreover, the hybrid-SIC scheme with task assignment outperforms the one without task assignment in the blue line. The one with task assignment have a wider lower-bound of the feasible range of the power allocation for case $\onesn=1$ in \eqref{OptP1}, which means that it can adopt the $\onesn=1$ case with smaller $P_{m,\phi}$. In addition, both hybrid SIC schemes has lower energy consumption than the OMA scheme.
In \fref{r5}, the energy consumption is presented as a function of the offloading data length. As the data length increases, the average energy consumption also grows. Our proposed hybrid-SIC scheme reduces the energy consumption significantly especially when the data length is large. Moreover, \fref{r6} reveals the energy consumption comparisons versus the maximum delay tolerance for $U_{n,\phi}$. With tight deadlines, the energy consumption of the hybrid-SIC scheme is much lower than OMA scheme, and more portion of data is processed locally to save energy compared to the fully offloading curve.

\begin{figure}
	\centering
	\includegraphics[scale=0.6]{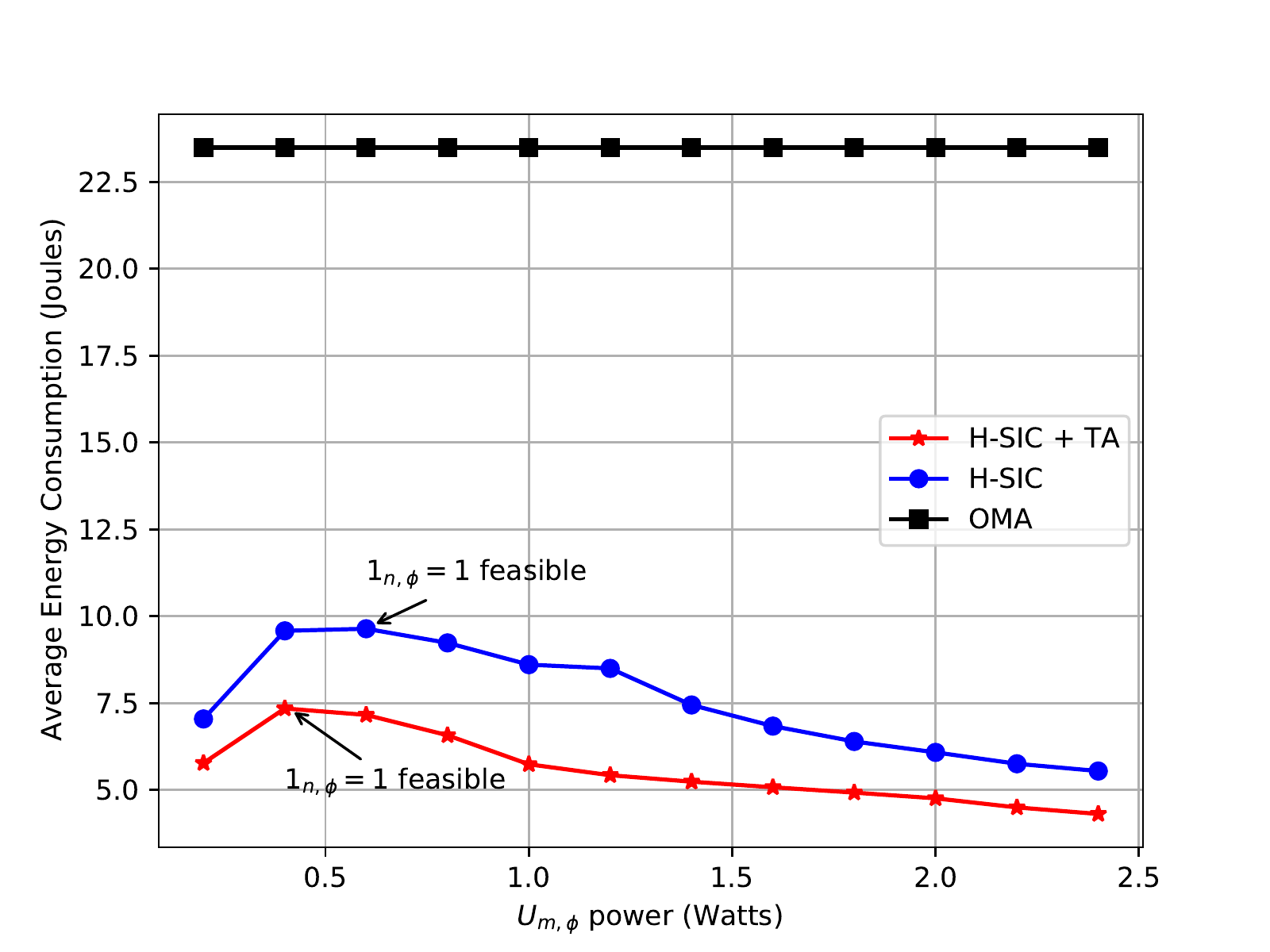}
	\caption{Average energy consumption versus training episodes with different numbers of users. }
	\label{r4}
\end{figure}
\begin{figure}
	\centering
	\includegraphics[scale=0.6]{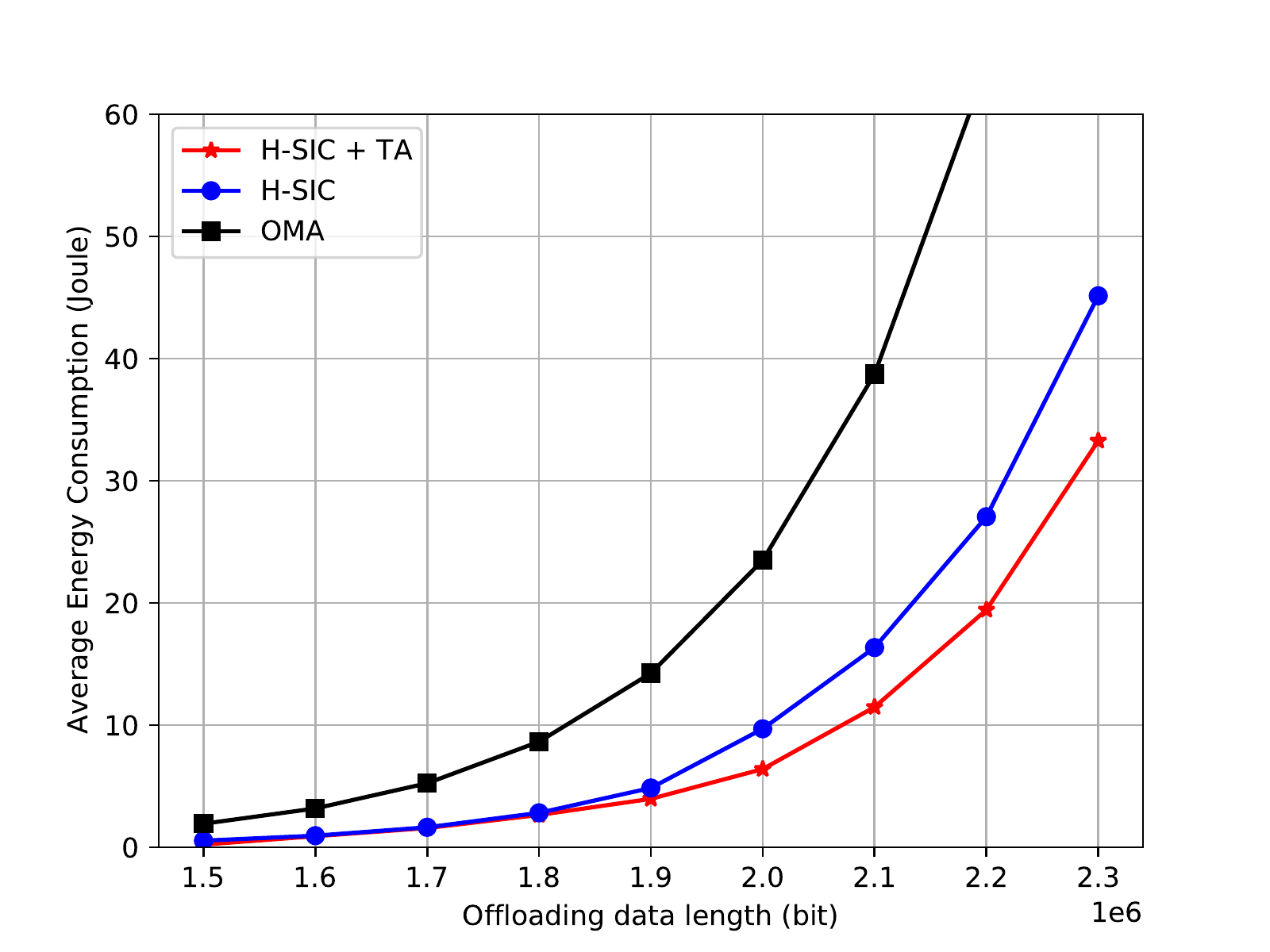}
	\caption{Average energy consumption versus training episodes with different numbers of users. }
	\label{r5}
\end{figure}
\begin{figure}
	\centering
	\includegraphics[scale=0.6]{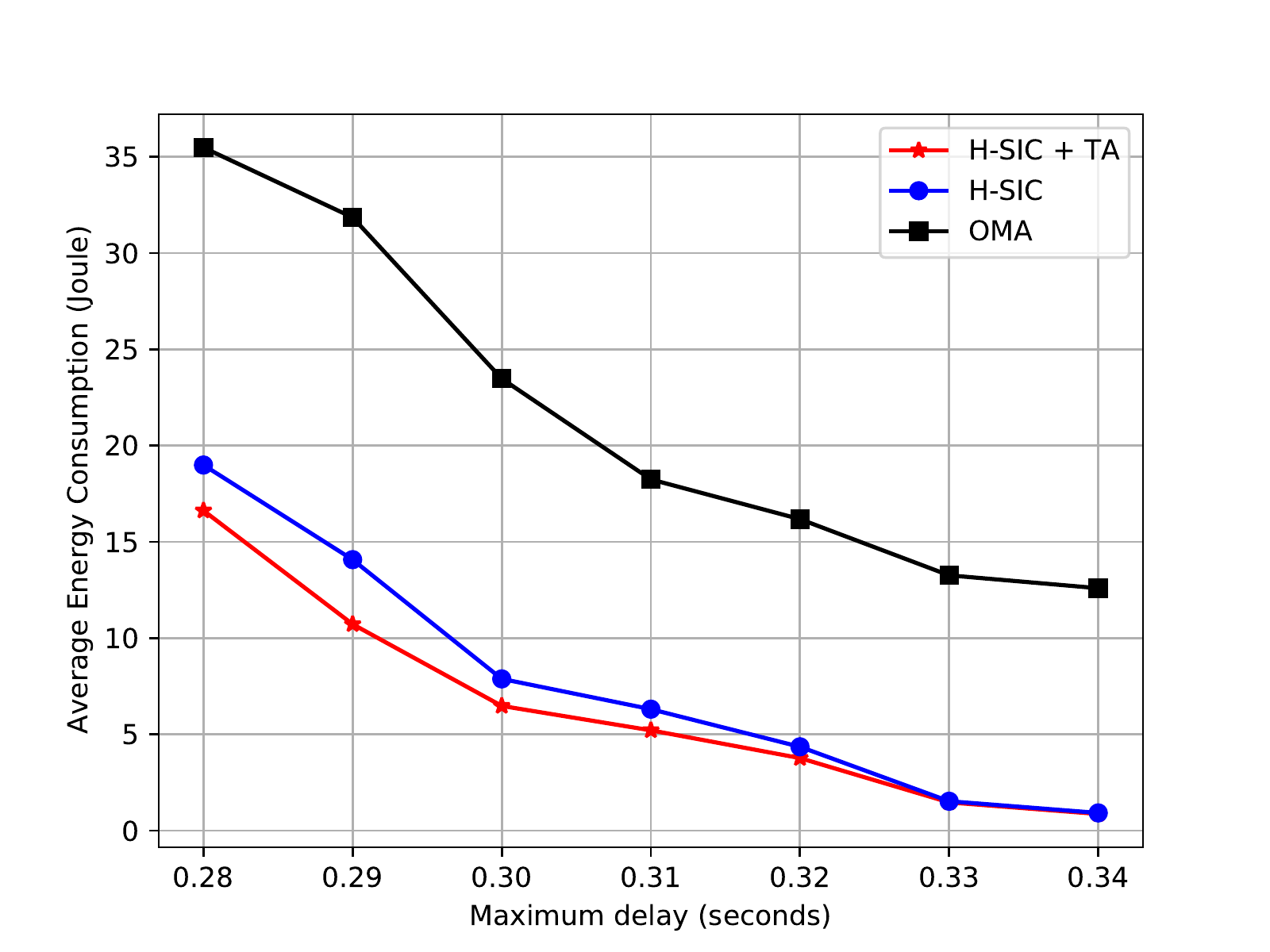}
	\caption{Average energy consumption versus training episodes with different numbers of users. }
	\label{r6}
\end{figure}

\section{Conclusion} \label{con}
This paper studied the resource allocation problem for a NOMA-assisted MEC network to minimize the energy consumption of users' offloading activities.
The hybrid NOMA scheme has two duration during each time slot, in which NOMA is adopted to serve the both users simultaneously during the first time duration, and a dedicate time slot is scheduled to offload the remaining part of the delay tolerable user solely by OMA. Upon fixing the user grouping, the non-convex problem was decomposed into three sub-problems including power allocation, time allocation and task assignment, which were all solved optimally by studying the convexity and monotonicity. The hybrid SIC scheme selects the SIC decoding order dynamically by a numerical comparison of the energy consumption between different decoding sequences. Finally, after solving those sub-problems, we proposed a DQN based user grouping algorithm to obtain the user grouping policy and minimize the long-term average offloading energy consumption. By comparing with various benchmarks, the simulation results proved the superiority of the proposed NOMA-MEC scheme in terms of energy consumption.

\appendix
	\subsection{Proof of Theorem 1}
	By fixing $t_{r,\phi}$ and $\beta_{\phi} $, the above problem in the case $\onesn =1 $ can be rewritten as:
\begin{subequations} \refstepcounter{prb} \label{AP1}
	\begin{align}
		\peqref{AP1}:
		\min_{\substack{\mathrm{P_{n,\phi},P_{r,\phi}}}}& \quad &&
		\frac{\kappa_{0}\left[C(1-\beta_{\phi})L\right]^{3}}{\left(\tau_{m,\phi}+t_{r,\phi}\right)^{2}} 
		+ \tau_{m,\phi}P_{n,\phi} + t_{r,\phi} P_{r,\phi}   \\
		\textrm{s.t.}&\quad &&  
		\tau_{m,\phi} B\ln\left(1+P_{n,\phi}\abs{h_{n,\phi}}^{2}  \right)   + t_{r,\phi} B\ln \left( 1+ P_{r,\phi}\abs{h_{n,\phi}}^2\right) \geq \beta_{\phi} L \label{AP1C1} \\
		&&& \tau_{m,\phi} B \ln \left(1+ \frac{P_{m,\phi}\abs{h_{m,\phi}}^2}{P_{n,\phi} \abs{h_{n,\phi}}^{2}+1} \right) \geq  L \\
		&&& P_{n,\phi} \geq 0, P_{r,\phi} \geq 0 \label{AP1C2} \\ 
	\end{align}
\end{subequations}
It is evident that the problem is convex, and by rearranging \eqref{AP1C2} as
\begin{equation}
	\begin{split}
		P_{n,\phi} \abs{h_{n,\phi}}^{2} - P_{m,\phi}\abs{h_{m,\phi}}^2 \left(e^{\frac{L}{B\tau_{m,\phi} }} -1\right)^{-1} + 1	 \leq 0,
	\end{split}
\end{equation}
the Lagrangian function can be obtained as follows:
\begin{equation}
	\begin{split}
		\mathcal{L}(P_{n,\phi},P_{r,\phi},\bm{\lambda})=&\frac{\kappa_{0}\left[C(1-\beta_{\phi})L\right]^{3}}{\left(\tau_{m,\phi}+t_{r,\phi}\right)^{2}} + \tau_{m,\phi}P_{n,\phi} + t_{r,\phi} P_{r,\phi} - \lambda_{1} P_{n,\phi}-\lambda_{2} P_{r,\phi} + \lambda_{3}\beta_{\phi}L \\
		&-\lambda_{3}\tau_{m,\phi} B\ln\left(1+P_{n,\phi}\abs{h_{n,\phi}}^{2}  \right)-\lambda_{3}t_{r,\phi} B\ln \left( 1+ P_{r,\phi}\abs{h_{n,\phi}}^2\right)\\
		&+\lambda_{4} \left( P_{n,\phi} \abs{h_{n,\phi}}^{2} - P_{m,\phi}\abs{h_{m,\phi}}^2 \left(e^{\frac{L}{B\tau_{m,\phi} }} -1\right)^{-1} + 1 \right),
	\end{split}
\end{equation}
where $\bm{\lambda} \triangleq [\lambda_{1},\lambda_{2},\lambda_{3},\lambda_{4}]$ are the Lagrangian multipliers.
The stationary conditions are given as

\begin{subequations}
	\begin{align}
		& 	\frac{\partial \mathcal{L}}{\partial P_{n,\phi}} = \tau_{m,\phi} - \lambda_{1}
		-\lambda_{3}\tau_{m,\phi} B\frac{\abs{h_{n,\phi}}^2}{P_{n,\phi}\abs{h_{n,\phi}}^2+1}
		+ \lambda_{4} \abs{h_{n,\phi}}^{2} = 0\\
		& \frac{\partial \mathcal{L}}{\partial P_{r,\phi}} = t_{r,\phi} - \lambda_{2} - \lambda_{3} t_{r,\phi}B \frac{\abs{h_{n,\phi}}^{2}}{P_{r,\phi}\abs{h_{n,\phi}}^{2} +1} = 0
	\end{align}
\end{subequations}
The Karush–Kuhn–Tucker (KKT) conditions \cite{cvx} can be obtained as

\begin{subequations}
	\begin{align}
		\beta_{\phi} L- \tau_{m,\phi} B\ln\left(1+P_{n,\phi}\abs{h_{n,\phi}}^{2}  \right)   
		- t_{r,\phi} B\ln \left( 1+ P_{r,\phi}\abs{h_{n,\phi}}^2\right) \leq 0\\
		P_{n,\phi} \abs{h_{n,\phi}}^{2} - P_{m,\phi}\abs{h_{m,\phi}}^2 \left(e^{\frac{L}{B\tau_{m,\phi} }} -1\right)^{-1} + 1	 \leq 0\\
		-P_{n,\phi} \leq 0, -P_{r,\phi} \leq 0\\
		\lambda_{i} \geq 0,\quad i\in \{1,2,3,4\}\\
		\lambda_{1}P_{n,\phi}=0\\
		\lambda_{2}P_{r,\phi}=0\\
		\lambda_{3}\beta_{\phi}L  -\lambda_{3}\tau_{m,\phi} B\ln\left(1+P_{n,\phi}\abs{h_{n,\phi}}^{2}  \right)-\lambda_{3}t_{r,\phi} B\ln \left( 1+ P_{r,\phi}\abs{h_{n,\phi}}^2\right)=0\\
		\lambda_{4} \left( P_{n,\phi} \abs{h_{n,\phi}}^{2} - P_{m,\phi}\abs{h_{m,\phi}}^2 \left(e^{\frac{L}{B\tau_{m,\phi} }} -1\right)^{-1} + 1 \right)=0\\
		\tau_{m,\phi} - \lambda_{1} -\lambda_{3}\tau_{m,\phi} B\frac{\abs{h_{n,\phi}}^2}{P_{n,\phi}\abs{h_{n,\phi}}^2+1}	+\lambda_{4} \abs{h_{n,\phi}}^{2}=0 \label{KKT1_tm} \\
		t_{r,\phi} - \lambda_{2} - \lambda_{3} t_{r,\phi}B \frac{\abs{h_{n,\phi}}^{2}}{P_{r,\phi}\abs{h_{n,\phi}}^{2} +1}
	\end{align}
\end{subequations}
The power allocation schemes can be obtained by different Lagrangian multipliers decisions as follows\\
\begin{itemize}
	\item	Hybrid NOMA: $\lambda_{1}=0$, $\lambda_{2}=0$, and $\lambda_{3} \neq 0$.
	\begin{itemize}
		\item	If $\lambda_{4}=0$: 
		\begin{equation}
			P_{n,\phi}^{*}=P_{r,\phi}^{*}=\abs{h_{n,\phi}}^{-2} \left( e^{\frac{\beta_{\phi} L }{B \left( \tau_{m,\phi}+t_{r,\phi} \right) } } -1\right)
		\end{equation}
		\begin{equation}
			P_{m,\phi}\abs{h_{m,\phi}}^{2} \geq e^{\frac{\beta_{\phi} L }{B \left( \tau_{m,\phi}+t_{r,\phi} \right) } }\left(e^{\frac{L}{B\tau_{m,\phi} }} -1\right)
		\end{equation}
		\item If $\lambda_{4}\neq 0$: 
		\begin{subequations}
			\begin{align}
				& 	P_{n,\phi}^{*} = \abs{h_{n,\phi}}^{-2}\left[P_{m,\phi}\abs{h_{m,\phi}}^{2} \left(e^{\frac{L}{B\tau_{m,\phi}}}-1\right)^{-1}-1  \right],\\
				&	P_{r,\phi}^{*} = \abs{h_{n,\phi}}^{-2} \left\{ e^{\frac{\beta_{\phi} L}{B t_{r,\phi}} - \frac{\tau_{m,\phi}}{t_{r,\phi}} \ln \left[P_{m,\phi}\abs{h_{m,\phi}}^{2}\left( e^{\frac{L}{B\tau_{m,\phi}}}-1 \right)^{-1} \right]  } -1 \right\},						
			\end{align}
		\end{subequations}			
		where $	e^{\frac{\beta_{\phi}L}{B\tau_{m,\phi}}}\left(e^{\frac{L}{B\tau_{m,\phi}}}-1 \right) \geq	P_{m,\phi}\abs{h_{m,\phi}}^{2} \geq e^{\frac{L}{B\tau_{m,\phi} }} -1$.
	\end{itemize}
	\item Pure NOMA: $\lambda_{1}=0$, $\lambda_{2} \neq 0$:
	\begin{subequations}
		\begin{align}
			&	P_{n,\phi} =\abs{h_{n,\phi}}^{-2}\left( e^{\frac{\beta_{\phi}L}{B\tau_{m,\phi}}}-1\right),\\
			&	P_{r,\phi} =0,
		\end{align}
	\end{subequations}
	where $P_{m,\phi}\abs{h_{m,\phi}}^{2} \geq e^{\frac{\beta_{\phi}L}{B\tau_{m,\phi}}}\left(e^{\frac{L}{B\tau_{m,\phi}}}-1 \right)$.
	\item OMA: $\lambda_{1}\neq 0$, $\lambda_{2} = 0$
	\begin{subequations}
		\begin{align}
			&	P_{n,\phi}=0, \\
			&	P_{r,\phi} = \abs{h_{n,\phi}}^{-2}\left( e^{\frac{\beta_{\phi} L }{t_{r,\phi}B}}-1 \right).
		\end{align}
	\end{subequations}
\end{itemize}

\subsection{Proof of Proposition 1}
The total energy consumption can be expressed as:
\begin{equation}
	\begin{split}
		E_{\mathrm{H1}} =&	\frac{\kappa_{0}\left[C(1-\beta_{\phi})L\right]^{3}}{\left(\tau_{m,\phi}+t_{r,\phi}\right)^{2}} + \tau_{m,\phi}\abs{h_{n,\phi}}^{-2}\left[P_{m,\phi}\abs{h_{m,\phi}}^{2} \left(e^{\frac{L}{B\tau_{m,\phi}}}-1\right)^{-1}-1  \right] \\
		&+ t_{r,\phi} \abs{h_{n,\phi}}^{-2} \left\{ e^{\frac{\beta_{\phi} L}{B t_{r,\phi}} - \frac{\tau_{m,\phi}}{t_{r,\phi}} \ln \left[P_{m,\phi}\abs{h_{m,\phi}}^{2}\left( e^{\frac{L}{B\tau_{m,\phi}}}-1 \right)^{-1} \right]  } -1 \right\},
	\end{split}
\end{equation}
where $a_{\phi}=\frac{\beta_{\phi} L -B\tau_{m,\phi}\ln \left[P_{m,\phi}\abs{h_{m,\phi}}^{2}\left( e^{\frac{L}{B\tau_{m,\phi}}}-1 \right)^{-1} \right]}{B}  $.
\begin{equation}
	\begin{split}
		\frac{\partial E_{\mathrm{H1}}}{\partial t_{r,\phi}}=-\frac{2\kappa_{0}\left[C(1-\beta_{\phi})L\right]^{3}}{\left(\tau_{m,\phi}+t_{r,\phi}\right)^{3}} + \abs{h_{n,\phi}}^{-2} \left( e^{\frac{a_{\phi}}{t_{r,\phi}}}- \frac{a_{\phi}}{t_{r,\phi}}e^{\frac{a_{\phi}}{t_{r,\phi}}} -1 \right).
	\end{split}
\end{equation}
Define $g(x) =  e^{\frac{a_{\phi}}{x}}- \frac{a_{\phi}}{x}e^{\frac{a_{\phi,1}}{x}} -1 $,
\begin{equation}
	g'(x)= \frac{a_{\phi,1}^{2}e^{\frac{a_{\phi,1}}{x}}}{x^{3}} \geq 0, \quad \forall x>0.
\end{equation}
Hence, $g(x)$ is monotonically increasing for $ x >0$, and 
$g(t_{r,\phi}) \leq g(\infty)=0 $.\\
Therefore, $\frac{dE_{\mathrm{H1}}}{d t_{r,\phi}} \leq 0 $, which is monotonically decreasing. Hence, the larger $t_{r,\phi} $ is scheduled, the less energy is consumed, and the optimal situation is when $t^{*}_{r,\phi}=\tau_{n,\phi}-\tau_{m,\phi} $.
\par
For the power allocation scheme in \eqref{OptP2}, the energy consumption is given as

\begin{equation}
	\begin{split}
		E_{\mathrm{H2}} =&	\frac{\kappa_{0}\left[C(1-\beta_{\phi})L\right]^{3}}{\left(\tau_{m,\phi}+t_{r,\phi}\right)^{2}} + (\tau_{m,\phi}+t_{r,\phi})\abs{h_{n,\phi}}^{-2} \left( e^{\frac{\beta_{\phi} L }{B \left( \tau_{m,\phi}+t_{r,\phi} \right) } } -1\right).
	\end{split}
\end{equation}
By obtaining the derivative with respect to $t_{r,\phi}$,

\begin{equation}
	\frac{\partial E_{\mathrm{H2}}}{\partial t_{r,\phi}}=-\frac{2\kappa_{0}\left[C(1-\beta_{\phi})L\right]^{3}}{\left(\tau_{m,\phi}+t_{r,\phi}\right)^{3}} + |h_n|^{-2} \left( e^{\frac{\beta_{\phi}L}{B\left(\tau_{m,\phi}+t_{r,\phi}\right)}} -t_{r,\phi} {\frac{\beta_{\phi}L}{B\left(\tau_{m,\phi}+t_{r,\phi}\right)}}e^{\frac{\beta_{\phi}L}{B\left(\tau_{m,\phi}+t_{r,\phi}\right)}} - 1\right).
\end{equation}

Define $g_{2}(x) \triangleq  e^{\frac{\beta_{\phi}L}{B\left(\tau_{m,\phi}+x\right)}} -x {\frac{\beta_{\phi}L}{B\left(\tau_{m,\phi}+x\right)}}e^{\frac{\beta_{\phi}L}{B\left(\tau_{m,\phi}+x\right)}} - 1 $,
and the derivative of $g_{2}(x)$ is

\begin{equation}
	g_2'(x)= \frac{\left(\beta_{\phi}L \right)^2 }{B^2\left(\tau_{m,\phi}+x\right)^3}  e^{\frac{\beta_{\phi}L}{B\left(\tau_{m,\phi}+x\right)}}  \geq 0, \quad \forall x>0.
\end{equation}
Thus, $g_2(x)$ is monotonically increasing for $ x >0$, and 
$g(t_{r,\phi}) \leq g(\infty)=0 $, which indicates $\frac{dE_{\mathrm{H2}}}{d t_{r,\phi}} \leq 0 $. Similar to the previous case, the energy function  is monotonically decreasing with respected to $t_{r,\phi}$, and the optimal time allocation is $t^{*}_{r,\phi}=\tau_{n,\phi}-\tau_{m,\phi} $.

\subsection{Proof to Proposition 2}
\begin{subequations} 
	\begin{align}
		\min_{\beta_{\phi}}&  &&
		\frac{\kappa_{0}\left[C(1-\beta_{\phi})L\right]^{3}}{\left(\tau_{m,\phi}+t^{*}_{r,\phi}\right)^{2}} 
		+ \tau_{m,\phi}P^{*}_{n,\phi} + t^{*}_{r,\phi} P^{*}_{r,\phi}   \\
		\textrm{s.t.}&\quad &&   0\leq \beta_{\phi} \leq 1. \label{P4C}
	\end{align}
\end{subequations}
The Lagrangian is given as

\begin{equation}
	\mathcal{L}(\beta_{\phi},\lambda_{5},\lambda_{6})=\frac{\kappa_{0}\left[C(1-\beta_{\phi})L\right]^{3}}{\left(\tau_{m,\phi}+t^{*}_{r,\phi}\right)^{2}} 
	+ \tau_{m,\phi}P^{*}_{n,\phi} + t^{*}_{r,\phi} P^{*}_{r,\phi} -\lambda_{5}\beta_{\phi}+\lambda_{6}\left(\beta_{\phi}-1\right)
\end{equation}
\begin{itemize}
	
	\item For the case $P_{m,\phi}=P_{n,\phi}$ in \eqref{OptP1}, the stationary condition is obtained as
	\begin{equation}
		\frac{\partial \mathcal{L}}{\partial \beta_{\phi}} = \frac{-3\kappa_{0}\left(CL\right)^{3} \left(1-\beta_{\phi}\right)^{2}}{\tau_{n,\phi}^{2}} + \frac{L}{B} \abs{h_{n,\phi}}^{-2}e^{\frac{\beta_{\phi}L}{B\tau_{n,\phi}}}-\lambda_{5}+\lambda_{6}=0.
	\end{equation}
	Therefore, the KKT conditions can be written as follows:
	\begin{subequations}
		\begin{align}
			-\beta_{\phi}\leq0\\
			\beta_{\phi}-1\leq0\\
			\lambda_{5}\beta_{\phi}=0\\
			\lambda_{6}\left(\beta_{\phi}-1\right)=0\\
			\frac{-3\kappa_{0}\left(CL\right)^{3} \left(1-\beta_{\phi}\right)^{2}}{\tau_{n,\phi}^{2}} + \frac{L}{B} \abs{h_{n,\phi}}^{-2}e^{\frac{\beta_{\phi}L}{B\tau_{n,\phi}}}-\lambda_{5}+\lambda_{6}=0 \label{kkt3_l}
		\end{align}
	\end{subequations}
	For $\beta_{\phi}>0$, $\lambda_{5}=\lambda_{6}=0$, and \eqref{kkt3_l} can be rewritten as
	\begin{equation}
		\frac{3\kappa_{0}\left(CL\right)^{3} \left(1-\beta_{\phi}\right)^{2}}{\tau_{n,\phi}^{2}}=\frac{L}{B} \abs{h_{n,\phi}}^{-2}e^{\frac{\beta_{\phi}L}{B\tau_{n,\phi}}}.
	\end{equation}
	Define $z_{1,\phi} = \frac{3\kappa_{0}BC^{3}L^{2}\abs{h_{n,\phi}}^{2}}{\tau_{n,\phi}^{2}}$, $z_{2,\phi} = \frac{L}{B\tau_{n,\phi}}$, and $b_{\phi} = \left(1-\beta_{\phi}\right)$, the optimal task assignment coefficient can be derived as
	\begin{equation}
		z_{1,\phi} b_{\phi}^{2}= e^{z_{2,\phi} \left(1-b_{\phi}\right)},
	\end{equation}
	\begin{equation}
		b_{\phi}=\frac{2}{z_{2,\phi}}\mathcal{W}\left(\frac{1}{2}z_{1,\phi}^{-\frac{1}{2}}z_{2,\phi}e^{\frac{z_{2,\phi}}{2}}\right).
	\end{equation}
The optimal task assignment ratio can be expressed as
	\begin{equation}
		\beta_{\phi}^{*}=1-b=1-\frac{2}{z_{2,\phi}}\mathcal{W}\left(\frac{1}{2}z_{1,\phi}^{-\frac{1}{2}}z_{2,\phi}e^{\frac{z_{2,\phi}}{2}}\right).
	\end{equation}
	
	\item For the case $P_{m,\phi} \neq P_{n,\phi}$ in \eqref{OptP2}:
	
	The stationary condition can be expressed as
	
	\begin{equation}
		\frac{\partial \mathcal{L}}{\partial \beta_{\phi}} = \frac{-3\kappa_{0}\left(CL\right)^{3} \left(1-\beta_{\phi}\right)^{2}}{\tau_{n,\phi}^{2}} + \abs{h_{n,\phi}}^{-2}\frac{L}{B}e^{-u}e^{\frac{\beta_{\phi}L}{B\left(\tau_{n,\phi}-\tau_{m,\phi}\right)}-u} -\lambda_{5}+\lambda_{6}=0,
	\end{equation}
	where $u_{\phi}= \frac{\tau_{m,\phi}}{\left(\tau_{n,\phi}-\tau_{m,\phi}\right)} \ln \left[P_{m,\phi}\abs{h_{m,\phi}}^{2}\left( e^{\frac{L}{B\tau_{m,\phi}}}-1 \right)^{-1} \right]$.
	
	\begin{subequations}
		\begin{align}
			-\beta_{\phi}\leq0\\
			\beta_{\phi}-1\leq0\\
			\lambda_{5}\beta_{\phi}=0\\
			\lambda_{6}\left(\beta_{\phi}-1\right)=0\\
			\frac{-3\kappa_{0}\left(CL\right)^{3} \left(1-\beta_{\phi}\right)^{2}}{\tau_{n,\phi}^{2}} + \abs{h_{n,\phi}}^{-2}\frac{L}{B}e^{-u_{\phi}}e^{\frac{\beta_{\phi}L}{B\left(\tau_{n,\phi}-\tau_{m,\phi}\right)}-u_{\phi}} -\lambda_{5}+\lambda_{6}=0 \label{TA2:l}
		\end{align}
	\end{subequations}
	For $\beta_{\phi}>0$, $\lambda_{5}=\lambda_{6}=0$, constraint $\eqref{TA2:l}$ can be rearranged as
	\begin{equation}
		\frac{3\kappa_{0}\left(CL\right)^{3} \left(1-\beta_{\phi}\right)^{2}}{\tau_{n,\phi}^{2}} = \abs{h_{n,\phi}}^{-2}\frac{L}{B}e^{-u_{\phi}}e^{\frac{\beta_{\phi}L}{B\left(\tau_{n,\phi}-\tau_{m,\phi}\right)}-u_{\phi}} ,
	\end{equation}
	\begin{equation}
		\frac{3\kappa_{0}B\abs{h_{n,\phi}}^{2}\left(CL\right)^{3}e^{2u_{\phi}} \left(1-\beta_{\phi}\right)^{2}}{\tau_{n,\phi}^{2}L}=e^{\frac{\beta_{\phi}L}{B\left(\tau_{n,\phi}-\tau_{m,\phi}\right)}}.
	\end{equation}
	Define $z_{1,\phi}=\frac{3\kappa_{0}B\abs{h_{n,\phi}}^{2}C^{3}L^{2}e^{2u_{\phi}} }{\tau_{n,\phi}^{2}} $, $z_{2,\phi}=\frac{L}{B\left(\tau_{n,\phi}-\tau_{m,\phi}\right)} $, the above equation can be rewritten as
	\begin{equation}
		z_{1,\phi}b_{\phi}^{2}=e^{z_{2,\phi}\left(1-b_{\phi}\right)},
	\end{equation}
	\begin{equation}
		b_{\phi}=\frac{2}{z_{2,\phi}}\mathcal{W}\left(\frac{1}{2}z_{1,\phi}^{-\frac{1}{2}}z_{2,\phi}e^{\frac{z_{2,\phi}}{2}}\right),
	\end{equation}
Hence the optimal task partition assignment ratio is:
	\begin{equation}
		\beta_{\phi}^{*}=1-b_{\phi}=1-\frac{2}{z_{2,\phi}}\mathcal{W}\left(\frac{1}{2}z_{1,\phi}^{-\frac{1}{2}}z_{2,\phi}e^{\frac{z_{2,\phi}}{2}}\right).
	\end{equation}	
\end{itemize}

\bibliographystyle{IEEEtran}
\bibliography{mybib}

\end{document}